\documentclass[12pt,reqno,a4paper]{amsart}

\usepackage[margin=3cm,footskip=1cm]{geometry}

\usepackage{amsmath}
\usepackage{amsfonts}
\usepackage{amsthm}
\usepackage{mathtools}
\DeclarePairedDelimiter\ket{\lvert}{\rangle}
\DeclarePairedDelimiter\bra{\langle}{\rvert}

\usepackage{enumerate}
\usepackage[latin1]{inputenc}

\usepackage{graphicx}
\usepackage{verbatim}

     \newcommand{\R}{{\mathbb{R}}}

\renewcommand{\d}{{\rm d}}

\newcommand\inp[2][]{#1 \langle #2#1\rangle}
\newcommand\parb[2][]{#1 \big ( #2#1\big )}
\newcommand\parub[2][]{#1 ( #2#1)}
\newcommand\parbb[2][]{#1 \Big ( #2#1\Big )}

\renewcommand{\exp}{{\rm exp}}
\newcommand{\mand}{\text{ and }}
\newcommand{\mfor}{\text{ for }}
\newcommand{\mforall}{\text{ for all }}

\newcommand{\vB}{{\mathcal B}}

\newcommand{\vE}{{\mathcal E}}

\newcommand{\vG}{{\mathcal G}}

\newcommand{\vH}{{\mathcal H}}

\newcommand{\vM}{{\mathcal M}}

\newcommand{\vO}{{\mathcal O}}

\newcommand{\vU}{{\mathcal U}}

\newcommand{\vS}{{\mathcal S}}

     \theoremstyle{plain}
     \newtheorem{thm}{Theorem}[section]
     \newtheorem{prop}[thm]{Proposition}
     \newtheorem{lemma}[thm]{Lemma}
      \newtheorem{cor}[thm]{Corollary}
     \theoremstyle{definition}
     
     \newtheorem{defn}[thm]{Definition}

     \newtheorem{remark}[thm]{Remark}
     \newtheorem{remarks}[thm]{Remarks}
\newtheorem*{remarks*}{Remarks}
\newtheorem*{remark*}{Remark}

     \numberwithin{equation}{section}
\title[Global solutions to the eikonal equation]{Global solutions to the eikonal equation}

\author{J. Cruz-Sampedro}
\address[J. Cruz-Sampedro]{Departamento de Ciencias B\'asicas \\
UAM-A \\ Av. San Pablo 180, Azcapotzalco  02200, Mexico}
\email{jacs@correo.azc.uam.mx}

\author{E. Skibsted}
\address[E. Skibsted]{Institut for  Matematiske
Fag \\
Aarhus Universitet\\ Ny Munkegade  8000 Aarhus C,
Denmark}
\email{skibsted@imf.au.dk}

\usepackage[all]{onlyamsmath}
\makeatletter
\newcommand\myparagraph[1]{\par\medskip\noindent\begingroup%
  \def\tagform@##1{\maketag@@@{\bfseries(\ignorespaces##1\unskip\@@italiccorr)}}%
  \textbf{\ignorespaces#1}%
  \endgroup%
  \enspace}
\makeatother

\begin{document}

\begin{abstract} We study structural stability of smoothness of
  the maximal solution to  the geometric eikonal equation on
  $(\R^d,G)$, $d\geq 2$.
This is within the framework of order zero metrics $G$. For a subclass
we show existence, stability as well as
precise asymptotics for derivatives of the solution.  These results
are  applicable for examples from Schr\"odinger operator
 theory.\end{abstract}

\maketitle

\section{Introduction and  results} \label{Introduction}
In this paper we investigate the existence of a smooth global
solution to the geometric
eikonal equation on the  Riemannian manifold $(\R^d,G)$, $d\geq 2$,
for a class of  metrics $G$. We are interested in the
so-called maximal solution  $S(x)$ constructed as the geodesic
distance from $x$ to
a given
fixed point $x_0$ (taken to be $x_0=0$). It is well-known that for some
$G$'s this function $S$ is smooth while for others this is not the
case, in fact  even $S\in C^1$ might be false. (At this point  the reader may consult
\cite{Li,CC} for  studies of ``generalized solutions'' to
related
Hamilton-Jacobi equations and further references.) Whence it is interesting to investigate the stability of
smoothness of $S$ under  perturbation of the metric. Other  issues  we will study
are bounds  and  the
asymptotic behaviour of derivatives of   a smooth $S$ at
infinity. This  is a general mathematical problem motivated by
 specific applications in  scattering theory, see
 \cite{ACH,Ba,IS,Sk}. 
 Most
likely 
it  is   relevant  for other specific problems too (possibly from 
control theory, geometric optics, etc) although this will not be
examined in this  paper.

To motivate our setup explained in further  details  below let us imagine a more general
situation: Let us consider a complete simply connected $d$-dimensional
manifold $(M,g)$, $d\geq 2$, and a point $x_0\in M$ for which the
exponential map $\Phi=\exp_{x_0}(1\cdot):TM_{x_0}\to \R^d$ is a
diffeomorphism. Then the pullback $G=\Phi^*g$ is a metric on
$TM_{x_0}$ that can be compared with the canonical one, $g_{x_0}$. Introducing
orthonormal coordinates we may identify $TM_{x_0}=\R^d$ and
$g_{x_0}(y,y)=|y|^2$ (the usual Euclidean metric). Upon doing this
identification the
Gauss lemma (see \cite[Theorem 1.8]{Ch}) implies that $G(x)$ (considered
as a matrix) has the  form
\begin{equation}
\label{ortdec22}
G(x)=P+P_\perp G(x) P_\perp,
\end{equation}
where $P$ denotes, in the Dirac notation,
the orthogonal projection $P=P(\hat x)=|\hat x\rangle\langle\hat x|$
parallel to $\hat x=x/|x|$ and $P_\perp=P_\perp(\hat x)=I-P$
the orthogonal projection onto $\{\hat x\}^\perp$.
Note that in this picture $S(x)=|x|$ and $G(0)=I$.

Due to the above  change
of framework we can in principle  reduce the study of  a general stability problem
to the one indicated above, i.e. for $(\R^d,G)$ only, in
fact for the unperturbed metric being of the form \eqref{ortdec22}. In
this paper we
shall consider families $\{G(\cdot)\}$ that are  of   order zero, see \eqref{condoz} and \eqref{ellcond} below for precise definition. Notice that this  class is naturally equipped with norms
giving precise meaning to the notion of ``perturbation''. Our main
result asserts that, under conditions, indeed upon perturbing a metric
of the form  \eqref{ortdec22} one obtains metrics, say denoted by $G_\epsilon(\cdot)$,
for which the geodesic distance to the origin, $S_\epsilon(\cdot)$, is smooth (more
precisely of class $C^l$ depending on conditions) and solves the
eikonal equation
\begin{equation*}
\nabla S_\epsilon G^{-1}_\epsilon\nabla S_\epsilon=1\;\mfor x \neq 0.
\end{equation*} Moreover introducing
\begin{equation*}
    s_\epsilon(x)=S_\epsilon(x)/|x|-1 \mfor x \neq 0
  \end{equation*} we have   bounds
  \begin{equation*}
    \sup_{x\in \R^d\setminus{\{0\}}}|x|^{|\alpha|}\left| \partial^\alpha
    s_\epsilon(x)\right|= o(\epsilon^0)\mfor |\alpha|\leq 2.
  \end{equation*}
Depending on conditions there are somewhat similar  bounds for higher order
derivatives.

Finally we use the change of frame in terms of  the exponential
mapping for the unperturbed metric, as explained above,  to solve eikonal equations of the form
\begin{equation*}
  |\nabla S_\epsilon(x)|^2=2(\lambda-V_\epsilon(x)) \mfor x\in \R^d\setminus\{0\},
\end{equation*} where $V_\epsilon$ is given by perturbing a negative radial
function (potential) obeying certain properties, see Section \ref{Examples}. Here the parameter $\lambda\geq
0$ plays in applications (Schr\"odinger operator theory) the role of
energy. The conventional way af constructing solutions is by a fixed
point method, cf. for example \cite{DS,Ho1,Is}. However to our knowledge this is
not doable for our examples.

This work is inspired by \cite{Ba} which has similar results for
perturbation of the special case $G=I$. As in   \cite{Ba} we shall use
the geometric/variational approach to define the maximal solution
 to the
eikonal equation, see \eqref{maxsol}. Many of our arguments are
however very
different from those of   \cite{Ba}. This came  out of
necessity to treat the present generality, see Remarks
\ref{remark:cond-main-results} \ref{item:10}),
\ref{remark:cond-main-resultppp}  and
\ref{remark:proof-theor-refthm:m} for   comments on this issue.

\subsection{Conditions and main results} \label{Conditions and main results}

Let $\mathcal S_d(\R)$, $d\geq 2$,  be the space of $d\times d$ symmetric matrices with components in $\R$ and, for $l\geq 2$, let
$\mathcal B^l(\R^d)$ be  the space of $C^l$ functions $G:\R^d\rightarrow \mathcal  S_d(\R)$ such that
 \begin{equation}
   \label{condoz}
\|G\|_l=\sup\{\langle x\rangle^{|\alpha|}\left| \partial^\alpha g_{ij}(x)\right|: x\in\R^d, |\alpha|\leq l, i,j=1,\dots,d\}<\infty,
 \end{equation}
where $G(x)=(g_{ij}(x))$ and $\langle x\rangle=(1+|x|^2)^{1/2}$. The space $\mathcal B^l(\R^d)$ endowed with the norm defined in \eqref{condoz} is a Banach space. Let $\mathcal M=\mathcal M(\R^d)$ be the set of $G\in \mathcal B^l(\R^d)$ for which there are positive constants $a$ and $b$  satisfying
 \begin{equation}
   \label{ellcond}
   a|y|^2\leq y G(x)y\leq b|y|^2,\qquad \quad x,y\in\R^d.
 \end{equation}
The set  $\mathcal M$  is open in $ \mathcal B^l(\R^d)$, and its
elements will be referred to as \emph{metrics of order zero}.
We denote  by  $\mathcal H$ the
Sobolev space   $ (H^1_0(0,1))^d$ with the norm
\begin{equation*}
  \|h\|^2=\inp{h,h}=\int_0^1|\dot h(s)|^2ds.
\end{equation*}
For  $G\in \mathcal M$ we consider the energy functional  $E:\R^d\times
   {\mathcal H}\rightarrow \R$  given, for
   $(x,\kappa)\in \R^d\times \mathcal H $ and $y(s)=sx+\kappa(s)$,  by
\begin{equation}
\label{ener}
E(x,\kappa)=\int_0^1 \dot y(s)G(y(s))\dot y(s)ds.
\end{equation}

\enlargethispage{2em}
Define for  any $G\in \vM$ a non-negative  function $S$ by
 \begin{equation}
   \label{maxsol}
S^2(x)=\inf\{E(x,\kappa):
\kappa\in \mathcal H\};\;x\in\R^d,
\end{equation}
where $E$ is as in \eqref{ener}.

Let
\begin{equation}
\label{paths}
\mathcal E_x=\{y\in (H^1(0,1))^d: y(s)=sx+\kappa(s), \kappa\in\mathcal H\};\;x\in\R^d.
\end{equation}

In agreement with usual convention \cite{Mi}, we say that
$\gamma\in\mathcal E_x$ is a \emph{geodesic of $G$ emanating from $0$
  with value $x$ at time one} if $\partial_\kappa E(x,\kappa)=0$ as an
element of the dual space $\mathcal H'$ of $\mathcal H$. Geodesics can
emanate from other points in $\R ^d$. The general definition of a
\emph{geodesic of $G$} (\cite{Ch, Mi}) may be taken to be an orbit
$s\to \gamma(s)\in \R^d$ solving a certain second order differential
equation, see \eqref{eq:3}. These solutions can be extended to
globally defined solutions, cf. the Hopf Rinow theorem
\cite[Theorem~1.10]{Ch}.  By the Riesz lemma we can identify $\mathcal
H'$ and $\mathcal H$, as well as the set of bounded quadratic forms on
$\mathcal H$ by the set of bounded self-adjoint operators on $\mathcal
H$.

We shall consider  two sets of metrics of order zero. The first one is
the following.

\begin{defn} \label{def:regmet} Let  $\vU$ be the
  subclass of order zero metrics  $G$   that  satisfy:
\begin{enumerate}[1)]

\item\label{item: Condition 1}
 For every $x\in \R^d$ the metric $G$ has a unique geodesic $\gamma_x(s)=sx+\kappa_x(s)$ in $\mathcal E_x$.

\item \label{item: Condition 2}
There exists $c>0$ independent of $x\in\R^d$ such that if
$\kappa=\kappa_x$ is given as in \ref{item: Condition 1})
\begin{equation}
\label{positivity}
\langle \partial^2_\kappa E(x,\kappa)h,h\rangle\geq c\|h\|^2, \qquad\quad h\in \mathcal H.
\end{equation}
\end{enumerate}
\end{defn}
\begin{remarks*}
\begin{enumerate}[1)]
\item \label{item:11}The Hessian appearing in \eqref{positivity} is
  given by
  \begin{align}
    \MoveEqLeft\langle \partial^2_\kappa E(x,\kappa) h_1,h_2 \rangle \nonumber\\
    &=\int_0^1(2\dot h_1G\dot h_2 +2\dot y\nabla G\cdot h_1 \dot h_2
    +2\dot y\nabla G \cdot h_2 \dot h_1 + \dot y (\nabla^2
    G;h_1,h_2)\dot y)ds;\label{hessianGb}
\end{align}  here   $\nabla G\cdot h$ stands for the matrix
$(\nabla g_{ij}\cdot h) $. Let us for completeness of presentation remark that for
all $G\in \vM$ satisfying \ref{item: Condition 1}) there exists $C>0$ independent of $x\in\R^d$ such that
\begin{equation*}
|\langle \partial^2_\kappa E(x,\kappa)h,h\rangle|\leq C\|h\|^2, \qquad\quad h\in \mathcal H.
\end{equation*} (This follows readily from  \eqref{hessianGb},
\eqref{estmin} and \eqref{eq:11}.)

\item \label{item:12} For all $G\in \vU$ the corresponding exponential map
  $\Phi=\exp_0(1\cdot):T\R^d_0\to \R^d$ is a diffeomorphism, cf. \cite[Theorem
2.16]{Ch} and \cite[Theorem 14.1]{Mi}.
\end{enumerate}
\end{remarks*}
\begin{prop}
  \label{prop:cond-main-results} Let $G\in \vU$. The  non-negative function $S$
defined by \eqref{maxsol} is  of class $C^{l}$ on $\R^d\setminus\{0\}$ and satisfies the eikonal equation
 \begin{equation}
   \label{eikeq}
\nabla S G^{-1}\nabla S=1\;\mfor x \neq 0.
\end{equation}

Furthermore, there exists $C>0$ such that
\begin{equation}
  \label{eq:1}
  \sup_{|x|\geq 1}\langle x\rangle^{\min(|\alpha|-1,\,|\alpha|/2)}\left| \partial^\alpha
    S(x)\right|\leq C \mforall |\alpha|\leq l.
\end{equation}
\begin{remark*} One might suspect that
  \eqref{eq:1} can be replaced by the stronger bounds
\begin{equation}
  \label{eq:1b}
  \sup_{|x|\geq 1}\langle x\rangle^{|\alpha|-1}\left| \partial^\alpha
    S(x)\right|\leq C \mforall |\alpha|\leq l.
\end{equation} In general this is an open problem. On the other hand
\eqref{eq:1}  appears  natural if the class of metrics is enlarged by
replacing \eqref{condoz} by
\begin{equation*}
\widetilde{\|G\|_l}=\sup\{\langle x\rangle^{\min (|\alpha|,1+|\alpha|/2)}\left| \partial^\alpha g_{ij}(x)\right|: x\in\R^d, |\alpha|\leq l, i,j=1,\dots,d\}<\infty.
 \end{equation*} In this situation indeed the analogue version of
 Proposition \ref{prop:cond-main-results} holds true (showned by the same proof).
  \end{remark*}
\end{prop}

Our second set of metrics of order zero is
given as follows.
\begin{defn} \label{def:ortdecomp}   Let
$\mathcal O$ be the subset  of order zero metrics  $G$  obeying:

\begin{enumerate}[1)]

\item\label{item: Condition 1b}  For  $x \neq 0$
\begin{equation}
\label{ortdec}
G(x)=P+P_\perp G(x) P_\perp,
\end{equation}
where $P$ denotes, in the Dirac notation,
the orthogonal projection $P=P(\omega)=|\omega\rangle\langle\omega|$ parallel to $\omega=\hat x=x/|x|$ and $P_\perp=P_\perp(\omega)=I-P$ the orthogonal projection onto $\{\omega\}^\perp$.

\item \label{item: Condition 2b}
There exists $\bar c >0$ such that for  $x \neq 0$
\begin{equation}
\label{ortdec2}
P_\perp \parb {G(x) +2^{-1}x \cdot \nabla G(x)} P_\perp\geq \bar c  P_\perp G(x) P_\perp.
\end{equation}
\end{enumerate}

\end{defn}
Note that since $G$ is continuous at $x=0$ Definition
\ref{def:ortdecomp} \ref{item: Condition 1b}) implies that
$G(0)=I$. Moreover it follows from \eqref{ortdec2} that $\bar c \in
]0,1]$. The simplest example of a metric of order zero satisfying
\eqref{ortdec} and \eqref{ortdec2} is $G=I$, the $d\times d$ identity
matrix. In Section \ref{Examples} we provide other examples of metrics
of order zero satisfying \eqref{ortdec} and~\eqref{ortdec2}. We may
refer to \eqref{ortdec2} as a \emph{convexity property} since, given
the orthogonal decomposition \eqref{ortdec} the estimate is equivalent
to the geometric Hessian bound
\begin{equation}
\label{ortdec2kk}
\nabla^2 S(x)^2\geq 2\bar c  g(x).
\end{equation} Here we use the conventional metric notation $g$ rather
than the matrix notation $G$ and (with \eqref{ortdec}) $S(x)=|x|$. Written in this way
the condition  \eqref{ortdec2}  clearly becomes \emph{geometrically
invariant} which  a priori is a desirable property. However in
computations we shall only use \eqref{ortdec2}.

In terms of the  non-negative function $S$  given by
\eqref{maxsol} (for  any $ G\in \vM$) let
  \begin{equation*}
    s(x)=S(x)/|x|-1 \mfor x \neq 0.
  \end{equation*}
 Our first main result is:

 \begin{thm}\label{thm:main result}
Let $\vU, \mathcal O\subseteq \vM$ be given by Definitions
\ref{def:regmet}
and \ref{def:ortdecomp}. There exists a neighbourhood
 $\widetilde \vO\subseteq \vM$ of   $\vO$ such that:
\begin{enumerate}[i)]
\item
\label{item:1} $\widetilde \vO
 \subseteq \mathcal \vU$; that is the set  $\vO$ is a subset of the
  interior of   $\mathcal \vU$.
\item\label{item:2}  Let $G\in
  \vO$ be given. Then there exist $\epsilon_{0},C>0$ such that for all  $\tilde G\in
  \vM$ with $\|\tilde G-G\|_l\leq\epsilon_0$ not only  $\tilde G\in
 \widetilde \vO$  but also
  \begin{subequations}
    \begin{align}
  \label{eq:errestza}
  \sup_{x\in \R^d\setminus{\{0\}}}|
    s(x)|&\leq C \|\tilde G-G\|_l,\\
\label{eq:errestzb}
  \sup_{x\in \R^d\setminus{\{0\}}}|x|\left| \partial^\alpha
    s(x)\right|&\leq C \|\tilde G-G\|^{3/4}_l \mfor |\alpha|= 1,\\
\label{eq:errestz}
  \sup_{x\in \R^d\setminus{\{0\}}}|x|^2\left| \partial^\alpha
    s(x)\right|&\leq C \|\tilde G-G\|_l^{1/2} \mfor |\alpha|= 2.
\end{align}
\end{subequations}
\end{enumerate}
\end{thm}

 \begin{remarks} \label{remark:cond-main-results}
   \begin{enumerate}[1)]
   \item \label{item:9} One might presume that $\vU$ is open in
     $\mathcal M$.  However if true at all this is a hard
     problem. Similarly (seemingly a softer problem) one might presume
     that if we drop the condition \eqref{ortdec2} of the Definition
     \ref{def:ortdecomp} and call this bigger class $\vO_1$, then
     $\vO_1$ is a subset of the interior of $U$. (Note that $
     \vO_1\subseteq U$ due to Lemma~4.1.). Even this problem appears
     to be difficult. In our approach we use \eqref{ortdec2} crucially
     to obtain good control of perturbed geodesics uniformly in
     $x$. Given the lower bound in \eqref{ellcond} the condition
     \eqref{ortdec2} is a somewhat weak assumption.
   \item \label{item:10a} The two constants $\epsilon_0, C>0$ can be taken
     as locally bounded functions of $(\|G\|_l,a, \bar c)\in\nobreak \R_+^3$
     where the entries $ \|G\|_l$, $a$, and $\bar c$ are defined by
     \eqref{condoz}, \eqref{ellcond} and \eqref{ortdec2} for the
     metric $G$, respectively. Obviously this statement for $l\geq 3$
     follows from the assertion for $l=2$.
\item \label{item:10} For perturbations of the
Euclidean metric $G=I$
  one can replace the powers to the right in
  \eqref{eq:errestza}--\eqref{eq:errestz} by the more natural factor $\|\tilde G-G\|_l $, see
  \cite{Ba}. However the techniques of \cite{Ba} are  not applicable in
  our case, see Remark \ref{remark:proof-theor-refthm:m}
  for some elaboration at this point. We do not know if this improvement
  is possible in our more general case.
   \end{enumerate}

 \end{remarks}

The estimate \eqref{eq:errestzb} is a consequence, by interpolation,
of the bounds \eqref{eq:errestza} and \eqref{eq:errestz},
cf. \cite[proof of Lemma 7.7.2]{Ho}. We shall use \eqref{eq:errestza}
and the following  weaker version of
\eqref{eq:errestzb} in the proof of \eqref{eq:errestz}:
\begin{equation}
  \label{eq:errestzbB}
  \sup_{x\in \R^d\setminus{\{0\}}}|x|\left| \partial^\alpha
    s(x)\right|\leq C \|\tilde G-G\|^{1/2}_l \mfor |\alpha|= 1.
\end{equation}

 Our second  main result supplements Theorem \ref{thm:main result}
 \ref{item:2}). It reads

\begin{thm}\label{thm:main-result2}  Suppose $l\geq 3$.  Let $G\in
  \vO$  and $r>0$ be given. Then there exist $\epsilon_{0},C>0$ such that for all  $\tilde G\in
  \vM$ with $\|\tilde G-G\|_l\leq\epsilon_0$ and $|\alpha|\leq l$
  \begin{subequations}
    \begin{align}
  \label{eq:1bhhhh}
  \sup_{|x|\geq r}\langle x\rangle^{|\alpha|-1}\left| \partial^\alpha
    S(x)\right|&\leq C,\\
\sup \,\langle x\rangle^{|\alpha|-2}\left| \partial^\alpha
    S^2(x)\right|&\leq C. \label{eq:1bhhhhg}
\end{align}
  \end{subequations}
\end{thm}
\begin{remark}\label{remark:cond-main-resultppp}
  If $l\geq 4$  we  can use Theorems \ref{thm:main result} and
  \ref{thm:main-result2} and interpolation to show that
  \begin{equation}
    \label{eq:75bbb}
    \sup_{|x|\geq r}|x|^{|\alpha|}\left| \partial^\alpha
    s(x)\right|\leq C_r \|\tilde G-G\|_l^{2^{1-|\alpha|}}\mfor
  2\leq |\alpha|\leq l-1.
  \end{equation} A slightly improved bound for $|\alpha|=3$ was proved in
  \cite{Ba} (in the setting of   \cite{Ba}) under the assumption that  $l=3$. We do not
  need estimates like \eqref{eq:75bbb} for  $|\alpha|> 2$ in  our
  applications \cite{IS,Sk}. On the other hand Theorem
  \ref{thm:main-result2} for   $l=3$ is indeed  important in these  applications.
\end{remark}

This paper is organized as follows: In Section~\ref{The minimization
  problem} we study the minimization problem~\eqref{maxsol}. Using
standard arguments we show the existence of a minimizer and some basic
properties of any such minimizer. In Section \ref{Proof of Proposition
  ref prop:cond-main-results} we show Proposition
\ref{prop:cond-main-results}. The proof is based on the implicit
function theorem, and the somewhat lengthy scheme for proving the
bounds \eqref{eq:1} is used again in Section \ref{Proof of Theorem
  main result2} to establish the improved bounds of Theorem
\ref{thm:main-result2}. The proof of Theorem~\ref{thm:main result}
\ref{item:1}), given in Section \ref{Proof of Theorem thm:main
  result}, is based on an analysis yielding dynamical control of
perturbed geodesics. We obtain sufficient control to be able to deduce
the uniqueness of the energy minimizer from a result from global
analysis. In Section \ref{Proof of Theorem thm:main result2} we show
Theorem~\ref{thm:main result} \ref{item:2}) by using various explicit
computations in combination with results from Section \ref{Proof of
  Theorem thm:main result}. The proof of
Theorem~\ref{thm:main-result2}, given in Section \ref{Proof of Theorem
  main result2}, is based on functional analysis arguments for Hardy
spaces taylored to the problem in hand. The first part of the proof is
devoted entirely to setting this up abstractly.  The second part is
devoted to verification of conditions, and as indicated above, this
involves a scheme from Section \ref{Proof of Proposition ref
  prop:cond-main-results}. In Section~\ref{Examples} we present
examples from Schr\"odinger operator theory.

\section{The minimization problem} \label{The minimization problem}

In this section we study some basic properties for metrics  $G\in\vM$.
Since $S(x)=|x|$ when $G(x)=I$
then,  from \eqref{ellcond} and \eqref{maxsol}, we have
\begin{equation}
   \label{natest}
a|x|^2\leq S^2(x)\leq b|x|^2,
\end{equation}
for all $x\in \R^d$.

 \begin{lemma}\label{lemma:existence}
   Let $G\in\vM$  and  $E$ as in \eqref{ener}.
 Then for every
 $x\in \R^d$ there exists  $\kappa \in\mathcal H$  such that
\begin{equation}
  \label{minatt}
  S^2(x)=E(x,\kappa).
\end{equation}
Moreover, if
 $\gamma$  is a geodesic of $G$ emanating from $0$
   with value  $x$ at time one then
\begin{equation}
\label{dconlaw}
 \dot \gamma(s)G(\gamma(s))\dot \gamma(s)=\int_0^1\dot\gamma(t)G(\gamma(t))\dot \gamma(t) dt,\qquad s\in[0,1].
\end{equation}
In particular, if $\gamma$ is a minimizer of \eqref{ener}  then   for all $s\in[0,1]$ we have
\begin{equation}
 \label{conlaw}
 \dot \gamma(s)G(\gamma(s))\dot \gamma(s)=S^2(x),
\end{equation}
 and
\begin{equation}
 \label{estmin}
 \frac{a}{b}|x|^2\leq |\dot \gamma(s)|^2 \leq\frac{b}{a} |x|^2\qquad
\hbox{and }\quad \frac{a}{b}|sx|^2\leq |\gamma(s)|^2\leq\frac{b}{a} |sx|^2.
\end{equation}
 \end{lemma}

 \begin{proof}
\myparagraph{Existence of a minimizer.}
To establish the existence of a  minimizer   it suffices
to show that, for any fixed $x\in \R^d$, the functional  $E(x,\cdot)$
is weakly lower semicontinuous on $\mathcal H$, and that
\begin{equation}\label{eq:29}
\lim_{\|\kappa\|\to\infty}E(x,\kappa)=\infty.
\end{equation}
 Note that \eqref{eq:29} follows from the estimate,  $y(s):=sx+\kappa(s)\in\mathcal E_x$,
\begin{equation*}
E(x,\kappa) =\int_0^1\dot y(s)G(y(s))\dot y(s) ds\geq a \int_0^1|\dot y(s)|^2ds.
\end{equation*}

 To  prove the weak lower  semicontinuity of $E$ in the second
 variable we let
 $\{\kappa_n\}$ be a sequence in $\mathcal H$ that converges  weakly
 to $\kappa\in \mathcal H$ and write  $y(s)=sx+\kappa(s)$ and
 $y_n(s)=sx+\kappa_n(s)$. Then
\begin{align*}
E (x,\kappa_n)- E(x,\kappa)={}&\int_0^1 (\dot y_n-\dot y)G(y)(\dot y_n-\dot y)ds
 -2\int_0^1\dot yG(y)(\dot y-\dot y_n )ds \nonumber\\
 &-\int_0^1\dot y_n(G(y)-G(y_n))\dot y_n ds.
\end{align*}
Since  the first integral is non-negative,  the second
goes to zero as $n\to\infty$ (because  $y_n-y$ converges
weekly to zero)  and the third integral goes to zero too
 (by  the compact embedding of $\mathcal H$ into $( C(0,1))^d$)  we find  that
\begin{equation*}
  \liminf_{n\to\infty} E(x,\kappa_n)\geq E(x,\kappa).
\end{equation*}

\myparagraph{The conservation law.}
 Since $G$ is of class $C^l$ so is $E(x,\kappa)$ and
for all $h\in\mathcal H$ we have
\begin{equation}
\label{graden}
\langle \partial_\kappa E(x,\kappa),h\rangle=\int_0^1[\dot y \nabla G\cdot h\dot y +2\dot hG\dot y]ds.
\end{equation}
  Moreover, if $\gamma$ is a geodesic of  $G$ then for all  $h\in\mathcal H$
\begin{equation}
\label{derener}
\int_0^1 (\dot \gamma \nabla G\cdot h\dot \gamma+2\dot hG\dot \gamma)ds=0.
\end{equation}
 Integrating by parts we obtain
\begin{equation}
\label{dualcond}
\int_0^1 [\dot \gamma \nabla G\cdot h \dot \gamma
-2h(\nabla G\cdot {\dot\gamma}\dot \gamma +G\ddot\gamma) ]dt=0.
\end{equation}
  Thus, in the  space  $\mathcal H'$  we have
\begin{equation}
\label{derg}
2e_k G\ddot\gamma=\dot\gamma \nabla G\cdot e_k\dot\gamma-2e_k\nabla G\cdot \dot\gamma \dot\gamma
\end{equation}
for each $k=1,\dots,d$, where $e_k$ is the $k^{\rm th}$ element of the
standard basis of $\R^d$; $e_1=(1,0,\dots, 0)$, $e_2=(0,1,\dots, 0)$,
etc.  This implies that
\begin{subequations}
 \begin{equation}
\label{eq:negeod}
\ddot\gamma=\tfrac{1}{2}G^{-1}w,
\end{equation}
where the   $k^{\rm th}$ component $w_k$ of the function $w:[0,1]\rightarrow \R^d$ is given by
\begin{equation}
\label{eq:kcomp}
w_k=\dot\gamma
\nabla G\cdot e_k\dot\gamma -2e_k \nabla G\cdot \dot\gamma\dot \gamma.
\end{equation}
\end{subequations} Letting  $e_\bullet$ denote the  ordered   basis  $e_1,\dots, e_d$
of $\R^d$ we can
write \eqref{eq:negeod} and \eqref{eq:kcomp} more compactly as
\begin{equation}
  \label{eq:3}
  \ddot \gamma=\ddot \kappa=2^{-1}G^{-1}\dot \gamma \nabla G\cdot e_\bullet\dot \gamma-G^{-1}\nabla G\cdot \dot \gamma \dot \gamma.
\end{equation}

Multiplying \eqref{derg} by
$\dot\gamma_k$ and adding over $k$  yields
\begin{equation*}
 \dot\gamma \nabla G\cdot \dot\gamma\dot \gamma
 -2\dot\gamma \nabla G\cdot\dot\gamma\dot\gamma -2\dot\gamma
 G\ddot\gamma=0,
\end{equation*}
from which we obtain
\begin{equation}
\label{geodener}
\frac{d}{ds}\left(\dot\gamma(s) G(\gamma(s))\dot\gamma(s)\right)=0,\qquad\quad s\in[0,1],
\end{equation}
and therefore \eqref{dconlaw} and \eqref{conlaw} follow.

\myparagraph{The estimates.} Suppose $\gamma$ is a minimizer of
\eqref{ener}.  Combining \eqref{ellcond}, \eqref{natest}
and~\eqref{conlaw} we obtain the estimates for $|\dot\gamma(s)|$
stated in \eqref{estmin}.

 Furthermore,  using  the upper bound  for  $|\dot\gamma(s)|$ we find that
\begin{equation*}
  |\gamma (s)|\leq \int_0^s|\dot \gamma(\tau) |d\tau\leq \sqrt{\frac{b}{ a}}|sx|,\qquad \hbox{for all $s\in[0,1]$.}
\end{equation*}
Suppose now that for  some $s_1\in (0,1]$ we have $|\gamma(s_1)|^2<(a/b) |s_1x|^2$, then defining
$y(s)=(s/s_1)\gamma (s_1)$ for $0\leq s\leq s_1$ and $y(s)=\gamma (s)$ for $s_1\leq s\leq 1$ we have
\begin{equation*}
  E(y)=\int_0^{s_1} s_1^{-2}\gamma(s_1) G(y(s))\gamma(s_1)ds
+\int_{s_1}^1 \dot\gamma(s) G(\gamma(s))\dot\gamma(s)ds.
\end{equation*}
Thus, using \eqref{ellcond}, \eqref{natest},  \eqref{conlaw}, and our assumption on $s_1$ we find that
\begin{equation*} 
  E(y)\leq \frac{b |\gamma(s_1)|^2}{ s_1}+(1-s_1)S^2(x)
  <a |x|^2 s_1+(1-s_1)S^2(x)\leq S^2(x),
\end{equation*}
which is impossible by the choice of $\gamma$. Thus \eqref{estmin}
holds and the proof is complete.

 \end{proof}

\section{Proof of Proposition
  \ref{prop:cond-main-results}} \label{Proof of Proposition ref prop:cond-main-results}

Clearly Proposition
  \ref{prop:cond-main-results}  follows  from

\begin{prop} \label{prop:soleikeq}
Let $G\in\vU$ and write the unique  geodesic for   $G$  with endpoint
$x$ as $\gamma_x(s)=sx+\kappa_x(s)$.
\begin{enumerate}[i)]
\item \label{item:3} The map $\R^d\ni
x\rightarrow \kappa_x\in \mathcal H$
 is of class $C^{l-1}$.
\item \label{item:4} The non-negative function $S$ in
 \eqref{maxsol}
 obeys
\begin{equation}
\label{sol1ee}
S^2(x)=\int_0^1\dot\gamma_x(s)G(\gamma_x(s))\dot\gamma_x(s) ds,
\end{equation}
 and it is a $C^{l}$ solution to the eikonal equation \eqref{eikeq}.

\item \label{item:5a}  There are bounds
  \begin{equation}
    \label{eq:6}
    \|\partial^\alpha\kappa\|\leq C\langle x\rangle^{1/2-\min(|\alpha|-1/2,\,|\alpha|/2)}\mforall |\alpha|\leq l-1.
  \end{equation}

\item \label{item:5}

For any $r>  0$
\begin{equation}
  \label{eq:1cc}
  \sup_{|x|\geq r}\langle x\rangle^{\min(|\alpha|-1,\,|\alpha|/2)}\left| \partial^\alpha
    S(x)\right|\leq C_r \mforall |\alpha|\leq l.
\end{equation}

\item \label{item:5bvc}  In \ref{item:5a}) and \ref{item:5})  the
  constants $C$ and $C_r$ can be taken as  locally
bounded functions
of $(\|G\|_l,a, c)\in \R_+^3$ and $(\|G\|_l,a, c,r)\in \R_+^4$,
respectively. Here the entries $ \|G\|_l$, $a$ and $c$ are defined by
\eqref{condoz}, \eqref{ellcond} and \eqref{positivity},  respectively.
\end{enumerate}

\end{prop}
\begin{proof} \myparagraph{Re \ref{item:3}).}  The statement  follows from our
  assumption that
  $G\in \vU$, the representation
  \eqref{graden} and the implicit function theorem. Note that indeed
  the  classical  implicit function theorem given for example in
  \cite[Theorem 15.1]{De} (see possibly also \cite[Theorem C.7]{Ir})
  can be applied  to the equation  $\partial_\kappa E(x,\kappa)=0$
  in a small neighborhood of  any fixed $(x_0,\kappa_{x_0})\in (\R^d,\vH)$. The
  unique solution   is a map $x\rightarrow \tilde\kappa_x$   of class
  $C^{l-1}$  from  a neighborhood of  $x_0\in \R^d$ to $\mathcal
  H$. Since the corresponding  geodesic $\tilde\gamma_x$,
  $\tilde\gamma_x(s)=sx+\tilde\kappa_x$, by our uniqueness assumption
  coincides with $\gamma_x$ we deduce that
  $\kappa_x=\tilde\kappa_x$. Whence  also $x\to \kappa_{x}$
 is of class $C^{l-1}$.

 \myparagraph{Re \ref{item:4}).} Clearly \eqref{sol1ee} is a
 consequence of Lemma \ref{lemma:existence} and the uniqueness of
 geodesics.  It follows from \ref{item:3}) and \eqref{sol1ee} that
 $S^2$ is of class $C^{l-1}$ on $\R^d$. In particular, since $S(x)>0$
 for $x\in \R^d\setminus\{0\}$, $S$ is of class $C^{1}$ on
 $\R^d\setminus\{0\}$. If we write $\gamma$ instead of $\gamma_x$ in
 \eqref{sol1ee}, then using \eqref{dualcond}, \eqref{derg}, and the
 fact that $ \partial \gamma/\partial x_k=se_k+\partial
 \kappa/\partial x_k$, with $\kappa\in \mathcal H$, we have
\begin{align*}
  2S(x)\frac{\partial S(x)}{\partial x_k}&=\int_0^1\left(
    2\frac{\partial \dot\gamma}{\partial x_k}G\dot\gamma
    +\dot \gamma \nabla G\cdot\frac{\partial \gamma}{\partial x_k} \dot\gamma\right)ds\\
  &=\int_0^1(2e_kG\dot\gamma+s\dot\gamma\nabla G\cdot e_k\dot\gamma)ds\\
  &= \int_0^1 \left((2se_kG\dot\gamma)^\cdot+s(\dot\gamma \nabla G\cdot e_k\dot\gamma-(2e_kG\dot\gamma)^\cdot \right)ds\\
  &= 2e_k G(x)\dot\gamma(1).
\end{align*}
\enlargethispage{2em} 
Thus
\begin{equation}
\label{eq:rgrel}
G(x)\dot \gamma (1)=S(x)\nabla S(x),
\end{equation}
and using  \eqref{dconlaw} we obtain
\begin{equation*} S^2(x)=\dot\gamma(1) G(\gamma(1)) \dot\gamma(1) =S^2(x) \nabla
S(x)G(x)^{-1}\nabla S(x),
\end{equation*}
from which \eqref{eikeq} follows.

It remains to show that $S$ is of class $C^l$ on $\R^d\setminus\{0\}$:
From \eqref{eq:rgrel} we obtain
\begin{equation}
  \label{eq:2}
  \nabla S(x)=S(x)^{-1}G(x) (x+\dot\kappa_x(1)).
\end{equation} Whence it suffices to show that $\dot\kappa_x(1)$ is of class $C^{l-1}$.
 For
that let us note the representation
\begin{equation}
  \label{eq:4}
  \dot\kappa(1)=2\int^1_{1/2}\parbb{\dot\kappa(s)+\int^1_s\ddot\kappa(\sigma)
  d\sigma}ds.
\end{equation} The right hand side of \eqref{eq:4} is indeed of class $C^{l-1}$. Note
that in fact  it follows from \eqref{eq:4} that
\begin{equation}
  \label{eq:5}
  \partial_x^\alpha\dot\kappa(1)=2\int^1_{1/2}\parbb{\partial_x^\alpha\dot\kappa(s)+\int^1_s\partial_x^\alpha\ddot\kappa(\sigma)
  d\sigma}ds\mfor |\alpha|\leq l-1,
\end{equation} where the quantities $\partial_x^\alpha\dot\kappa$ and $\partial_x^\alpha\ddot\kappa$ are
well-defined
$\big (L^1(0,1)\big )^d-$valued function, cf.~\ref{item:3}) and~\eqref{eq:3}.

\myparagraph{Re \ref{item:5a}).} It suffices to show the bounds for
$|x|\geq 1$. We shall proceed by induction in $|\alpha|$. Note that
\eqref{eq:6} for $|\alpha|=0$ only follows from Lemma
\ref{lemma:existence} and the representation
$\kappa(s)=\gamma(s)-sx$. Suppose we know the bounds for $|\alpha|\leq
n-1$ then we need to show these for $|\alpha|= n$. So let $\alpha$
with $|\alpha|= n$ be given.

By repeated differentation of  the defining equation $\inp{\partial_\kappa
  E(x,\kappa),h}=0$ (for any $h\in \mathcal H$)  we obtain that
\begin{equation*}
  -\inp{\partial^2_\kappa E(x,\kappa)\partial^\alpha\kappa,h}
\end{equation*} is a sum of terms each one either of the form
\begin{subequations}
  \begin{equation}\label{eq:7}
  \inp{\partial^\zeta_x\partial_\kappa
  E(x,\kappa),h}=(\partial^\zeta_x\partial^{1+k}_\kappa
  E(x,\kappa);h); |\zeta|=n,\;k=0,
\end{equation} or of the form, for $k=1,\dots,n$,
\begin{gather}\label{eq:8}
  (\partial^\zeta_x\partial^{1+k}_\kappa
  E(x,\kappa);\partial^{\beta_1}\kappa,\dots, \partial^{\beta_k}\kappa,h);\\
|\zeta|+\sum_{j=1}^{k} |\beta_j|=|\alpha|=n\mand 1\leq |\beta_j|\leq
  n-1.\nonumber
\end{gather}
\end{subequations} Due to \eqref{positivity} it suffices to bound the
expressions in \eqref{eq:7} and \eqref{eq:8} as
\begin{equation}
  \label{eq:10}
  |\cdots|\leq C\langle x\rangle^{1/2-n/2}.
\end{equation}

From  \eqref{graden} we see that
\begin{equation*}
  (\partial^\zeta_x\partial^{1+k}_\kappa
  E(x,\kappa);h_1,\dots, h_k,h_{k+1})
\end{equation*} is an integral of a sum of   $(k+1)$-tensors  in
$g_1(s),\dots, g_{k+1}(s)$ where for each $j\leq k+1$ either $g_j(s)$ is
a component of $h_j(s)$
or $g_j(s)$ is a component of $\dot h_j(s)$. At most two factors have a
``dot superscript'', and if we include
factors of components of $\dot \gamma(s)$ and factors of components of
$\partial_{x_i}\tfrac{d}{ds}\{sx\}=e_i$, $i=1,\dots,d$,  the total number
of such factors is for all tensors exactly
two. Whence we  are motivated to group the  terms into three types that will be
considered separately below:
\begin{enumerate}[\bf A)]
\item There are no factors of components of $\dot \gamma(s)$.

\item There is one factor of a component of $\dot \gamma(s)$.

\item There are two factors of components of $\dot \gamma(s)$.
\end{enumerate}

Clearly these tensors involve factors of components of $\partial^\eta
G(\gamma(s))$ also. These are estimated as
\begin{equation}
  \label{eq:9}
  |\partial^\eta g_{ij}|\leq C
  |sx|^{-\sigma\min(|\eta|,\,1+|\eta|/2)};\;\sigma\in [0,1],
\end{equation} 
due to Lemma \ref{lemma:existence}. The singular power of $s$ in
\eqref{eq:9} (depending on the $\sigma$ at our disposal) needs to be
factorized into factors some of which need to be distributed to
factors of components of $ h_j(s)$ (if such components appear) and
then ``removed'' either by the Hardy inequality (removing a factor
$s^{-1}$)
\begin{subequations}
  \begin{equation}
  \label{eq:11}
  \int_0^1s^{-2}|\tilde h(s)|^2\,\d s\leq 4\|\tilde h\|^2,
\end{equation} or by  the estimate (removing a factor $s^{-1/2}$)
\begin{equation}
  \label{eq:12}
  |\tilde h(s)|\leq \sqrt s \|\tilde h\|.
\end{equation}
\end{subequations} Another factor of the singular power
of $s$ in \eqref{eq:9}  combines with factors of components of $\partial_{x_i}
(\gamma(s)-\kappa(s))=s e_i$, $i=1,\dots,d$.

 To summarize we need to look at the expressions
\eqref{eq:7} and \eqref{eq:8}. Let $h_j=\partial^{\beta_j}\kappa$
for $j\leq k$ and $h_{k+1}=h$. After  doing a complete expansion into terms (using the
product rule for differentiation) we need to bound each
resulting expression say $(F_{\zeta,k};h_1,\dots, h_{k+1})$. Recall from the
above discussion that any  such term is an integral of a
$(k+1)$-tensorial expression; the $j$'th factor is either a
component of    $h_j(s)$  or a component of $\dot
h_j(s)$. The treatment of
these  terms  is
divided into  various cases. For simplicity we assume below that
$k\geq1$. For  $k=0$ we can argue similarly (although the treatment
for  $k=0$ is simpler).

\myparagraph{Case A):} Notice that we need $|\eta|=|\zeta|+k-1$ in
\eqref{eq:9}. We distinguish between the following cases: \textbf{Ai)}
There occur a component of $\dot h_i(s)$ and a component of $\dot
h_j(s)$ (for some $i \neq j$). \textbf{Aii)} Exactly one factor of
component of $\dot h_j(s)$ occurs. \textbf{Aiii)} There is no factor
of component of $\dot h_j(s)$.

\myparagraph{Case Ai):} $\partial^\eta
G(\gamma(s))=s^{-|\zeta|}\partial^\zeta_x\partial_{\kappa(s)}^{\omega}
G(sx+\kappa(s));\;|\omega|=k-1$. We choose in \eqref{eq:9} $\sigma\in
[0,1]$ such that with the given value of $|\eta|$
\begin{equation}\label{eq:15}
  \sigma\min(|\eta|,\,1+|\eta|/2)=|\eta|/2=:K.
\end{equation}
Upon using the pointwise bound \eqref{eq:12} for $k-1$ factors, the
pointwise estimate
\begin{equation*}
  s^{-K}s^{|\zeta|}s^{(k-1)/2}\leq 1
\end{equation*} and the Cauchy Schwarz inequality we obtain
the bound
\begin{equation}
  \label{eq:13}
  |(F_{\zeta,k};h_1,\dots, h_{k+1})|\leq
  C
  |x|^{-K}\prod_{m=1}^{k+1} \|h_m\|.
\end{equation} By the induction hypothesis
\begin{equation*}
  \prod_{m=1}^{k} \|h_m\|\leq C\inp{x}^{k/2-\sum|\beta_j|/2}=C\inp{x}^{k/2-(n-|\zeta|)/2},
\end{equation*} which together with \eqref{eq:13} yields
\begin{equation}
  \label{eq:14}
  |(F_{\zeta,k};h_1,\dots, h_{k+1})|\leq
  C
  \inp{x}^{-K}\inp{x}^{k/2-(n-|\zeta|
)/2}\|h\|= C\inp{x}^{1/2-n/2}\|h\|.
\end{equation}
\myparagraph{Case Aii):} $\partial^\eta
G(\gamma(s))=s^{-|\zeta_1|}\partial^{\zeta_1}_x\partial_{\kappa(s)}^{\omega}
G(sx+\kappa(s))$; $|\zeta_1|=|\zeta|-1 ,\;|\omega|=k$. We choose
$\sigma$ as in \eqref{eq:15}. Upon
using  the bound \eqref{eq:11} for one factor and  the pointwise  bound
\eqref{eq:12} for $k-1$  factors we proceed as in the previous case
using now that
\begin{equation*}
  s^{-|\eta|/2}s^{|\zeta_1|}s^{1+(k-1)/2}\leq 1.
\end{equation*}

\myparagraph{Case Aiii):} $\partial^\eta
G(\gamma(s))=s^{-|\zeta_2|}\partial^{\zeta_2}_x\partial_{\kappa(s)}^{\omega}
G(sx+\kappa(s))$; $|\zeta_2|=|\zeta|-2,\;|\omega|=k+1$. We choose
$\sigma$ as in \eqref{eq:15}. Upon using the bound \eqref{eq:11} for
two factors and the pointwise bound \eqref{eq:12} for $k-1$ factors we
precced as in the first case using now that
\begin{equation*}
  s^{-|\eta|/2}s^{|\zeta_2|}s^{2+(k-1)/2}\leq 1.
\end{equation*}

\myparagraph{Case B):}
We need  $|\eta|=|\zeta|+k$ in
\eqref{eq:9}. We distinguish between the following cases: \textbf{Bi)} Exactly one factor of component of $\dot
h_j(s)$ occurs.
\textbf{Bii)}
There is  no factor of component of $\dot
h_j(s)$.

\myparagraph{Case Bi):} $\partial^\eta
G(\gamma(s))=s^{-|\zeta|}\partial^\zeta_x\partial_{\kappa(s)}^{\omega}
G(sx+\kappa(s));\;|\omega|=k$. We choose in \eqref{eq:9} $\sigma\in
[0,1]$ such that with the given value of $|\eta|$
\begin{equation}\label{eq:15b}
  \sigma\min(|\eta|,\,1+|\eta|/2)=1/2+|\eta|/2=:K.
\end{equation}
Upon using the bound \eqref{eq:11} for one factor and the pointwise
bound \eqref{eq:12} for $k-1$ factors, the pointwise estimate
\begin{equation*}
  s^{-K}s^{|\zeta|}s^{1+(k-1)/2}\leq 1
\end{equation*} and the Cauchy Schwarz inequality we obtain
the bound
\begin{equation}
  \label{eq:13bb}
  |(F_{\zeta,k};h_1,\dots, h_{k+1})|\leq
  C
  |x|^1|x|^{-K}\prod_{m=1}^{k+1} \|h_m\|.
\end{equation} (The factor $|x|^1$ comes from bounding
$|\dot\gamma(s)|\leq C|x|$, cf. Lemma \ref{lemma:existence}.) By the induction hypothesis
\begin{equation*}
  \prod_{m=1}^{k} \|h_m\|\leq C\inp{x}^{k/2-(n-|\zeta|)/2},
\end{equation*} yielding
\begin{equation}
  \label{eq:14bb}
  |(F_{\zeta,k};h_1,\dots, h_{k+1})|\leq
  C
  \inp{x}^{1-K}\inp{x}^{k/2-(n-|\zeta|
)/2}\|h\|= C\inp{x}^{1/2-n/2}\|h\|.
\end{equation}

\myparagraph{Case Bii):}
$\partial^\eta
G(\gamma(s))=s^{-|\zeta_1|}\partial^{\zeta_1}_x\partial_{\kappa(s)}^{\omega}
G(sx+\kappa(s));\;|\zeta_1|=|\zeta|-1,\;|\omega|=k+1$. We choose
$\sigma$ as in \eqref{eq:15b}. Upon
using  the bound \eqref{eq:11} for two  factors  and  the pointwise  bound
\eqref{eq:12} for $k-1$  factors we proceed as in  case Bi)
using now that
\begin{equation*}
  s^{-K}s^{|\zeta_1|}s^{2+(k-1)/2}\leq 1.
\end{equation*}

For the case   C) we need  \eqref{eq:9} with
$|\eta|=|\zeta|+k+1$.

\myparagraph{Case C):}
 $\partial^\eta
G(\gamma(s))=s^{-|\zeta|}\partial^\zeta_x\partial_{\kappa(s)}^{\omega}
G(sx+\kappa(s));\;|\omega|=k+1$. We choose in \eqref{eq:9} $\sigma\in [0,1]$ such
that with the given
 value of $|\eta|$
\begin{equation}\label{eq:15cc}
  \sigma\min(|\eta|,\,1+|\eta|/2)=1+|\eta|/2=:K.
\end{equation} Upon
using  the bound \eqref{eq:11} for two  factors and  the pointwise  bound
\eqref{eq:12} for $k-1$  factors, the
pointwise estimate
\begin{equation*}
  s^{-K}s^{|\zeta|}s^{2+(k-1)/2}\leq 1
\end{equation*} and the Cauchy Schwarz inequality we obtain
the bound
\begin{equation}
  \label{eq:13bbc}
  |(F_{\zeta,k};h_1,\dots, h_{k+1})|\leq
  C
  |x|^2|x|^{-K}\prod_{m=1}^{k+1} \|h_m\|.
\end{equation}  By the induction hypothesis
\begin{equation*}
  \prod_{m=1}^{k} \|h_m\|\leq C\inp{x}^{k/2-(n-|\zeta|)/2},
\end{equation*} yielding
\begin{equation}
  \label{eq:14bbc}
  |(F_{\zeta,k};h_1,\dots, h_{k+1})|\leq
  C
  \inp{x}^{2-K}\inp{x}^{k/2-(n-|\zeta|
)/2}\|h\|= C\inp{x}^{1/2-n/2}\|h\|.
\end{equation}

\myparagraph{Re \ref{item:5}).} The proof is by induction in
$|\alpha|$ and based on \eqref{eq:2}, \eqref{eq:5},  \eqref{eq:3} and
\ref{item:5a}). We notice that bound
\begin{equation}
  \label{eq:16}
  |\partial^\alpha\dot\kappa(1)|\leq C\langle x\rangle^{1/2-\min(|\alpha|-1/2,\,|\alpha|/2)}\mforall |\alpha|\leq l-1.
\end{equation} Indeed by repeated differentiation of \eqref{eq:3} it
suffices, due to \eqref{eq:5}, to bound the
 $L^1(1/2,1)$-norm of quantities of the form
 \begin{gather*}
   \partial^\eta g_{mn}\partial^{\beta_1}
   \dot\gamma_{j_1}(s)\,\partial^{\beta_2}\dot\gamma_{j_2}(s)
   \,\partial^{\omega_1}\gamma_{i_1}(s)\cdots\partial^{\omega_k}\gamma_{i_k}(s);\\
|\eta|=k+1,\;|\omega_i|\geq 1,\;\sum_{j=1,2}|\beta_j|+\sum_{i=1,\dots,k}|\omega_i|=|\alpha|.
 \end{gather*}
Using the bounds \eqref{eq:9}   and \ref{item:5a}) we obtain that
\begin{equation*}
  \|\partial^\alpha \ddot\kappa\|_{L^1(1/2,1)}\leq C\langle x\rangle^{1/2-\min(|\alpha|-1/2,\,|\alpha|/2)},
\end{equation*} yielding \eqref{eq:16}.

Now we apply the product rule to \eqref{eq:2} noticing that also
derivatives of the last factor $\dot \gamma(1)=x+\dot\kappa_x(1)$ obey
the bounds \eqref{eq:16}. Using the fact that $n\to \min
(n-\nobreak 2,(n-1)/2)$ is concave and the induction hypothesis it follows that
\begin{equation*}
  |\partial^\eta S(x)^{-1} |\leq C \inp{x}^{-\min(|\eta|-1,\,|\eta|/2)-2}.
\end{equation*} 
Since also
\begin{equation*}
  |\partial^\eta g_{mn}(x)|\leq C \inp{x}^{-\min(|\eta|,\,1+|\eta|/2)},
\end{equation*} 
indeed the product rule (and a little bookkeeping effort) completes
the induction argument.

\myparagraph{Re \ref{item:5bvc}).} This is obvious from  the above
proofs.

\end{proof}

\section{Proof of Theorem \ref{thm:main result}~\ref{item:1})} \label{Proof of Theorem thm:main result}

We embark on proving  the first  assertion of Theorem \ref{thm:main result}.
\subsection{Unperturbed case}\label{subsec:unperturbed-case}
Let $\mathcal O_1$ be  the subset of
  order zero metrics obeying Definition \ref{def:ortdecomp} \ref{item:
    Condition 1b}) and  let $\vU$ be  given as in Definition  \ref{def:regmet}.
\begin{lemma} \label{lemma:OinC}
$\mathcal O_1\subseteq \vU$.
\end{lemma}
\begin{proof}  Let $G\in\mathcal O_1$ be given.

\myparagraph{Re Definition \ref{def:regmet} \ref{item: Condition
    1}).}  A short calculation, using   \eqref{ortdec}, gives that for
all $h\in \R^d$
\begin{subequations}
  \begin{align}
  \label{gradG}
\nabla G\cdot h&=\nabla P\cdot h+
P_\perp\nabla G\cdot hP_\perp+\nabla P_\perp\cdot h G P_\perp+P_\perp
G\nabla P_\perp\cdot h;  \\
\label{gradP}
\nabla P\cdot h &=|\frac{P_\perp h}{|x|}\rangle\langle\omega|+|\omega\rangle\langle \frac{P_\perp h}{|x|}|\qquad \text{and} \qquad
\nabla P_\perp\cdot h=-\nabla P\cdot h.
  \end{align}
\end{subequations}

Now, using  \eqref{gradG} and
\eqref{gradP}  it can easily be verified  that
$\gamma(s)=sv$ is the solution  to  \eqref{eq:3}  that
satisfies $\gamma(0)=0$, and $\dot\gamma(0)=v$, for any given
$v\in\R^d$.  By standard uniqueness of the solution to an
initial value problem for an ordinary differential equation this $\gamma$ is the unique solution   to  \eqref{eq:3}  that
satisfies $\gamma(0)=0$, and $\dot\gamma(0)=v$.  Since $\gamma \in \mathcal E_x$ is a geodesic of $G$ if
and only if $\gamma$  satisfies \eqref{eq:3},  $\gamma(0)=0$, and $\gamma(1)=x$,  then $\gamma(s)= sx$ for $s\in[0,1]$, and thus $G$ satisfies Definition \ref{def:regmet} \ref{item: Condition 1}).

\myparagraph{Re Definition \ref{def:regmet} \ref{item: Condition 2}).}
We will use the representation \eqref{hessianGb}.  Since for all
$x\in\R^d$ the geodesic is $ \gamma(s)=sx$, a not very short but
elementary calculation, using \eqref{gradG}, \eqref{gradP}, and the
shorthand notation $\dot \gamma=d\gamma/ds$, $h_\perp =P_\perp h$ and
$\gamma\cdot\omega= \langle \gamma,\omega\rangle$, yields
\begin{align}
\MoveEqLeft 4\int_0^1 \dot \gamma\nabla G\cdot h\dot hds  \nonumber \\
&=4\int_0^1\bigg\{   \frac{h_\perp\cdot \dot h}{|\gamma|}  \dot \gamma\cdot\omega - \frac{h_\perp G\dot h_\perp}{|\gamma|}  \dot \gamma\cdot\omega\bigg\}ds\nonumber \\
&=4\int_0^1\big\{  h_\perp\cdot \dot h_\perp   -h_\perp G\dot h_\perp\big\} \frac{|\dot
  \gamma|}{|\gamma|}ds  \label{eq:31}
\end{align}
and
\begin{align}
\MoveEqLeft \int_0^1 \dot \gamma (\nabla^2 G;h,h)\dot \gamma ds \nonumber \\
&=\int_0^1 \dot \gamma\left\{    |\frac{\nabla P_\perp \cdot h}{|\gamma|}h\rangle\langle\omega|+|\omega\rangle
\langle h \frac{\nabla P_\perp\cdot h}{|\gamma|}| +2\nabla P_\perp\cdot h G \nabla P_\perp\cdot h \right\}\dot \gamma ds  \nonumber \\
&= 2\int_0^1\left\{-\frac{(\dot \gamma\cdot\omega)^2}{|\gamma|^2}|h_\perp|^2+\frac{(\dot \gamma\cdot \omega)^2}{|\gamma|^2}h_\perp Gh_\perp\right\}ds  \nonumber \\
&= 2\int_0^1 \big\{-|h_\perp|^2+h_\perp Gh_\perp\big\}\frac{|\dot
  \gamma|^2}{|\gamma|^2}ds.\label{eq:32}
\end{align}
Thus, using for the first and last equations  below that
\begin{equation}\label{eq:34}
  4h_\perp\cdot \dot h_\perp |\dot
  \gamma|/|
  \gamma|-2|h_\perp|^2|\dot
  \gamma|^2/|
  \gamma|^2=2(
h^2_\perp|\dot
  \gamma|/|
  \gamma|)^\cdot,
\end{equation}
 we find
\begin{align}
\MoveEqLeft \langle \partial^2_\kappa E(x,0) h,h \rangle\nonumber \\
&=2\int_0^1\bigg\{\dot hG\dot h - 2 h_\perp G\dot h_\perp\frac{|\dot
  \gamma|}{|\gamma|}   +  h_\perp Gh_\perp\frac{|\dot
  \gamma|^2}{|\gamma|^2}    \bigg\}ds\text{ (integrating  by
parts)}\nonumber \\
&=  2\int_0^1  \big\{  |P\dot h|^2+  s^2\big( s^{-1}h_\perp\big)^\cdot
G     \big( s^{-1} h_\perp\big)^\cdot \big\}ds\text{ (using }|\dot
  \gamma|/|\gamma|=s^{-1})\nonumber \\
&\geq   2\int_0^1 \big\{  |P\dot h|^2+  as^2\big|\big( s^{-1}h_\perp\big)^\cdot  \big|^2\big\}ds\text{ (by \eqref{ellcond})}\nonumber \\
&= 2\int_0^1 ( |P\dot h|^2+a | \parb{ h_\perp}^\cdot|^2)ds\text{ (integrating  by
parts)}.\label{eq:35}
\end{align}
We conclude (using the bound $a\leq 1$) that
\begin{equation}\label{eq:35n}
\langle \partial^2_\kappa E(x,0) h,h \rangle\geq c\|h\|^2; \quad
c\leq 2a\mand h\in \mathcal H.
\end{equation}
Hence $G$ satisfies Definition \ref{def:regmet} \ref{item: Condition 2}) and therefore $\mathcal O_1\subseteq \vU$.
\end{proof}

We will extend the above proof to the case of geodesics
  $\gamma\in \vE_x$ for  metrics ``near $\vO$'' a priori not knowing
  that geodesics are unique. To do this we need
  dynamical control of the geodesics and this will be provided under
  the
  additional condition
  \eqref{ortdec2}. First we discuss  the case of  $G\in \vO_1$
  (as in Lemma \ref{lemma:OinC}). Let $\gamma$ denote  any non-constant (maximal) geodesic for
  such a  metric (note that  we are here not assuming $ \gamma(0)=0$ but only
  the differential equation \eqref{eq:3}).
Introduce the observables
\begin{align}
  \label{eq:18}
  A&=\frac{\dot \gamma}{|\dot \gamma|_G}\cdot \hat \gamma\mand
  B=\hat {\dot \gamma}\cdot \hat \gamma;\\
|\dot \gamma|_G&=\sqrt{ \dot \gamma G(\gamma)\dot \gamma},\;\, \hat \gamma=\frac{\gamma}{|\gamma|},\;\hat {\dot \gamma}=\frac{\dot \gamma}{|\dot \gamma|}.\nonumber
\end{align}
Note that
\begin{equation}\label{eq:24}
  A=q\cdot \hat \gamma;\;q=\frac{G^{1/2}(\gamma)\dot \gamma}{|\dot
    \gamma|_G},
\end{equation} cf. \eqref{ortdec}.
In particular
  \begin{equation}
    \label{eq:17}A^2,B^2\leq 1.
  \end{equation}
\begin{lemma}
  \label{lemma:proof-theor-refthm:m-2} Suppose $G\in \vO_1$ and that
  $\gamma$ is a corresponding non-constant geodesic. Then
  \begin{equation}
    \label{eq:19}
    \dot A=\frac{|\dot \gamma|^2}{|\gamma|\,|\dot \gamma|_G}
    \inp{P_{\perp}(\hat \gamma)\hat {\dot \gamma},TP_{\perp}(\hat \gamma)\hat {\dot \gamma}};\;T=G(\gamma) +2^{-1}\gamma \cdot \nabla G(\gamma).
  \end{equation} In particular if $G\in \vO$
\begin{equation}
    \label{eq:19ab}
    \dot A\geq\bar c\frac{|\dot \gamma|_G}{|\gamma|}
    \parbb {1-A^2},
  \end{equation}  with $\bar c>0$ given by
  \eqref{ortdec2}.
\end{lemma}
  \begin{proof}
    Using \eqref{eq:3}, \eqref{gradG} and \eqref{gradP} we compute
    \eqref{eq:19}. Note that the denominator $|\dot \gamma|_G$ is
    preserved, cf. \eqref{geodener}.

    As for \eqref{eq:19ab} we use \eqref{eq:19}, \eqref{ortdec2} and
    \begin{equation}\label{eq:21}
   |P(\hat
    \gamma)\hat {\dot \gamma}|^2=1-|P_\perp(\hat
    \gamma)\hat {\dot \gamma}|^2=B^2=  \frac{|\dot \gamma|^2_G}{|\dot
      \gamma|^2}A^2.
    \qedhere
    \end{equation}
\end{proof}

\begin{lemma}
  \label{lemma:b_a}Let $G\in \vO_1$,  and let $A$ and $B$ be given by \eqref{eq:18}
for any $(\gamma,\dot \gamma)\in (\R^d\setminus\{0\})^2
$.
We have
\begin{subequations}
\begin{align}
\label{eq:23}
    b(1-B^2)&\geq  1-A^2\\
    \label{eq:23b}
    a^{-1}(1-A^2)&\geq  1-B^2.
  \end{align}
\end{subequations}
\end{lemma}
\begin{proof}  Using  \eqref{ortdec} and \eqref{eq:21} we can estimate
  \begin{equation}
    \label{eq:20}
\frac{|\dot \gamma|^2_G}{|\dot \gamma|^2}=|P\hat {\dot \gamma}|^2+|P_\perp G^{1/2}P_\perp \hat {\dot \gamma}|^2\leq    1+(b-1) (1-B^2),
  \end{equation} which in turn using \eqref{eq:17}, \eqref{eq:21} and
  the fact that $b\geq 1$ yields
  \begin{equation}
    \label{eq:22}
    1-B^2\geq 1-A^2-(b-1) (1-B^2).
  \end{equation} Obviously \eqref{eq:23} follows from \eqref{eq:22}.

Next we mimic the proof of \eqref{eq:23}. We have $1-A^2=
  1-\frac{|\dot \gamma|^2}{|\dot \gamma|_G^2}B^2 $, and letting $q=
  G^{1/2}\dot\gamma/|\dot \gamma|_G$ (as in \eqref{eq:24}) we estimate
  \begin{equation*}
\frac{|\dot \gamma|^2}{|\dot \gamma|_G^2}=|Pq|^2+|P_\perp G^{-1/2}q|^2\leq    1+(a^{-1}-1) (1-A^2),
  \end{equation*} cf.  \eqref{eq:20}. Whence the analogue
 of
  \eqref{eq:22} holds, and we conclude   \eqref{eq:23b}.
\end{proof}

We remark that  only \eqref{eq:23b} will  be needed. The estimate \eqref{eq:23} is
given only for completeness of presentation.

 \subsection{Perturbed case}\label{subsec:perturbed-case}
Now let $G\in \vO$ be given. We shall use the notation
$G_\epsilon$ to denote any metric of order zero obeing
$\|G_\epsilon-G\|_l\leq \epsilon$. The positive parameter $\epsilon$ is an order parameter which
we will take sufficiently small, say $\epsilon\leq \epsilon_0$,  in terms of quantities given by the
fixed unperturbed $G$. We shall use the observables $A$ and $B$ of
\eqref{eq:18} defined in terms of $G$ but now evaluated at $\gamma\to
\gamma_\epsilon$; clearly they are well-defined for
$\gamma_\epsilon(s), \dot\gamma_\epsilon(s)\neq 0$.
\begin{lemma}
  \label{lemma:perturbed-case} There exist $\epsilon_0>0$ and  $C_1,C_2,C_3>0$
  such that if $\|G_\epsilon-G\|_l\leq \epsilon\leq \epsilon_0$ and
  $\gamma_\epsilon$ is any non-constant  geodesic for the metric $G_\epsilon$
   with $\gamma_\epsilon(0)=0$, then
  \begin{enumerate}[i)]
  \item\label{item:6}  $\gamma_\epsilon(s)\neq 0$ for all $s>0$.

  \item \label{item:7} $B=B(\gamma_\epsilon(s),
    \dot\gamma_\epsilon(s))\geq 1-\epsilon C_1$  for all $s>0$.
  \item \label{item:8}There are bounds
  \begin{subequations}
   \begin{align}
 \label{estminbb1}
 ( 1-\epsilon C_2)|\dot \gamma_\epsilon(0)|^2&\leq |\dot \gamma_\epsilon(s)|^2 \leq ( 1+\epsilon C_2)|\dot \gamma_\epsilon(0)|^2,\\( 1-\epsilon C_3)|s\dot \gamma_\epsilon(0)|^2&\leq |\gamma_\epsilon(s)|^2\leq( 1+\epsilon C_3)|s\dot \gamma_\epsilon(0)|^2.\label{estminbb2}
\end{align}
  \end{subequations}

  \end{enumerate}
\end{lemma}
\begin{proof}
  \myparagraph{Re \ref{item:6}).} Let $\gamma_\epsilon$ be any such
  geodesic. Being non-constant implies $\dot \gamma_\epsilon(0)\neq
  0$. Indeed note at this point that the quantity $\dot
  \gamma_\epsilon G_\epsilon(\gamma_\epsilon)\dot \gamma_\epsilon$ is
  constant. Moreover, using also \eqref{ellcond} and the fact that
  $G(0)=I$,
  \begin{subequations}
    \begin{align}
      \label{eq:28p}
      a|\dot \gamma_\epsilon(s)|^2&\leq |\dot
      \gamma_\epsilon(s)|_{G(\gamma_\epsilon(s))}^2 \leq b|\dot \gamma_\epsilon(s)|^2,\\
      \label{eq:28}
      ( 1-\epsilon C)|\dot \gamma_\epsilon(0)|^2&\leq |\dot
      \gamma_\epsilon(s)|_{G(\gamma_\epsilon(s))}^2 \leq ( 1+\epsilon
      C)|\dot \gamma_\epsilon(0)|^2,\\
      \label{eq:28b}
      ( 1-\epsilon C)b^{-1}\dot \gamma_\epsilon(0)|^2&\leq |\dot
      \gamma_\epsilon(s)|^2 \leq ( 1+\epsilon C)a^{-1}|\dot
      \gamma_\epsilon(0)|^2.
    \end{align}
  \end{subequations} 
  We shall use \eqref{eq:28p} and \eqref{eq:28} later in the proof
  while \eqref{eq:28b} is stated just for completeness (note that
  \eqref{estminbb1} is stronger than \eqref{eq:28b}).

  Whence the observables $A$ and $B$ are well-defined at this geodesic
  on an interval of the form $]0,s_0[$. We need to show that $s_0$ can
  be taken arbitrarily large. Suppose not, then we let $s_0$ be the
  first positive nullpoint, $\gamma_\epsilon(s_0)=0$, and we need to
  find a contradiction.  For that it suffices show the bound
  \begin{equation}
    \label{eq:25}
    A(s)\geq 1-\epsilon \bar C\mforall s\in]0,s_0[,
  \end{equation} where  the constant $\bar C>0$ depends only on $G$
  (i.e. it is independent of $\epsilon$, $G_\epsilon$ and
  $\gamma_\epsilon$). In particular we can assume  that   $\epsilon
  \bar C<1$. It then follows
  that
  $A,B>0$. Whence by the computation
  \begin{equation}
    \label{eq:26}
    \tfrac{d}{ds}|\gamma_\epsilon|^2=2\inp{\gamma_\epsilon,\dot\gamma_\epsilon}=2|\gamma_\epsilon|\,|\dot\gamma_\epsilon|B>0,
  \end{equation} we see that $s\to |\gamma_\epsilon(s)|$ is
  increasing  yielding the contradiction and whence showing~\ref{item:6}).

  It remains to show \eqref{eq:25}. We compute the time-derivative of
  $A$ and find the following extension of \eqref{eq:19ab}.
\begin{equation}
    \label{eq:19abc}
    \dot A\geq \bar c\frac{|\dot \gamma_\epsilon|_{G(\gamma_\epsilon(s))}}{|\gamma_\epsilon|}
    \parb {1-A^2}-\epsilon \tilde C\frac{|\dot \gamma_\epsilon|}{|\gamma_\epsilon|}.
  \end{equation} Note that
  \begin{equation*}
    \tfrac{d}{ds} |\dot\gamma_\epsilon|_G^{-1} =\frac{\tfrac{d}{ds} \{\dot\gamma_\epsilon\parb{G_\epsilon-G}\dot\gamma_\epsilon\}}{2|\dot\gamma_\epsilon|_G^{3}},
  \end{equation*} and whence due to \eqref{eq:28p} that
\begin{equation*}
    |\tfrac{d}{ds} |\dot\gamma_\epsilon|_G^{-1} |\leq \epsilon \breve C|\gamma_\epsilon|^{-1}.
  \end{equation*} The constant  $ \breve C$ contributes to the
  constant $ \tilde C$ in \eqref{eq:19abc}. Another term of the form
  of the last term in  \eqref{eq:19abc}
  comes from comparing the right hand side of the geodesic equation \eqref{eq:3} for
  $\gamma_\epsilon$ with  the same  expression replacing
  $G_\epsilon(\gamma_\epsilon)\to G(\gamma_\epsilon)$.

Using \eqref{eq:28p} again we can simplify \eqref{eq:19abc} as
\begin{equation}
    \label{eq:19abcs}
    \dot A\geq \bar
    c\sqrt a\frac{|\dot \gamma_\epsilon|}{|\gamma_\epsilon|}
    \parb{1-A^2-\epsilon K};\;K=\tfrac{\tilde C}{\bar c\sqrt a}.
  \end{equation} Combining \eqref{eq:19abcs} and the fact that  $\lim_{s \to 0^+}A(s)=
  1$ (here once more using that $G(0)=\nobreak I$), we obtain   \eqref{eq:25} for any
  $\bar C>K/2$ provided $\epsilon$ is  chosen small enough (precisely we need
  $2\bar C-K>\epsilon {\bar C}^2$).

  \myparagraph{Re \ref{item:7}).}
We write $B=1-(1-B)$ and then use \eqref{eq:23b} and \eqref{eq:25} (we
use \eqref{eq:25}  with $s_0=\infty$) to estimate
  \begin{equation}
    \label{eq:27}
 (1+B)B\geq (1+B)-\epsilon \bar Ca^{-1}(1+A)\geq 1+B-\epsilon 2\bar Ca^{-1}.
  \end{equation} We subtract $B$ from \eqref{eq:27} and obtain
  \begin{equation*}
    B\geq \sqrt{1-\epsilon 2\bar Ca^{-1}}\geq 1-\epsilon \bar
    Ca^{-1}-\epsilon^2 C\geq
    1-\epsilon C_1;\;C_1:=2\bar Ca^{-1}.
  \end{equation*}

  \myparagraph{Re \ref{item:8}).} Using
  \ref{item:7}) we obtain, cf. \eqref{eq:20} and \eqref{eq:27},
  \begin{equation*}
     \parb{1-\epsilon 2C_1}|\dot
  \gamma_\epsilon(s)|^2\leq |\dot
  \gamma_\epsilon(s)|_{G(\gamma_\epsilon(s))}^2.
  \end{equation*} Also, using \eqref{eq:25}, we have
  \begin{equation*}
   (1-\epsilon\bar C)|\dot
  \gamma_\epsilon(s)|_{G(\gamma_\epsilon(s))} \leq |\dot
  \gamma_\epsilon(s)|_{G(\gamma_\epsilon(s))} A\leq |\dot
  \gamma_\epsilon(s)|,
  \end{equation*} yielding
\begin{equation*}
     |\dot
  \gamma_\epsilon(s)|_{G(\gamma_\epsilon(s))}^2 \leq   (1-\epsilon\bar C)^{-2}|\dot
  \gamma_\epsilon(s)|^2.
  \end{equation*}
Using the shown two-sided estimates in  combination with \eqref{eq:28} we obtain
  \eqref{estminbb1}.

Using
\begin{equation*}
   \tfrac{d}{ds} |\gamma_\epsilon|=|\dot \gamma_\epsilon|B,
\end{equation*} cf. \eqref{eq:26}, we obtain \eqref{estminbb2} from \eqref{estminbb1} and
\ref{item:7}) by integration.
\end{proof}

We have the following version of Definition \ref{def:regmet}
\ref{item: Condition 2}) for the perturbed metrics.
 \begin{lemma}
   \label{lemma:2_perturbed-case} There exists $\epsilon_0>0$ (possibly
   smaller than the $\epsilon_0$ of Lemma \ref{lemma:perturbed-case}) such
   that if  $\|G_\epsilon-G\|_l\leq \epsilon\leq \epsilon_0$, $x\in\R^d$ and
  $\gamma_\epsilon$ is any geodesic for the metric $G_\epsilon$
  emanating from $0$ with value $x$ at time one, then writing
  $\gamma_\epsilon(s)=sx+\kappa(s)$, denoting by $E_\epsilon$ the
  energy functional \eqref{ener} with $G$ replaced by  $G_\epsilon$ and using the
  positive number $a$ given by \eqref{ellcond} for the metric $G$,
\begin{equation}
\label{positivitycc}
\langle \partial^2_\kappa E_\epsilon(x,\kappa)h,h\rangle\geq a\|h\|^2, \qquad\quad h\in \mathcal H.
\end{equation}
 \end{lemma}
 \begin{proof} Note that for $x=0$ the geodesic $\gamma_\epsilon$ is
   unique due to Lemma \ref{lemma:perturbed-case} \ref{item:6}), it is
   given by $\gamma_\epsilon=0$. By \eqref{hessianGb} $
   \partial^2_\kappa E_\epsilon(0,0)=2G_\epsilon(0)$, so obviously
   \eqref{positivitycc} holds in this case provided  $a\leq 2a-2\epsilon
   d\;(\leq 2G_\epsilon(0))$.

   Suppose next that $x\neq 0$. We can mimic the proof of
   Lemma~\ref{lemma:OinC} using Lemma~\ref{lemma:perturbed-case}. By
   letting $s=1$ in \eqref{estminbb2} we obtain from \eqref{estminbb1}
   and \eqref{estminbb2} that there exist $C_4, C_5>0$ such that
   \begin{subequations}
   \begin{align}
 \label{estminbb1q}
 ( 1-\epsilon C_4)|x|^2&\leq |\dot \gamma_\epsilon(s)|^2 \leq ( 1+\epsilon C_4)|x|^2,\\( 1-\epsilon C_5)|sx|^2&\leq |\gamma_\epsilon(s)|^2\leq( 1+\epsilon C_5)|sx|^2.\label{estminbb2q}
\end{align}
  \end{subequations} A consequence of \eqref{estminbb1q} and
  \eqref{estminbb2q} is that the quantity $|\dot
  \gamma_\epsilon|/|\gamma_\epsilon|$ appearing when  we try to
  repeat the proof of Lemma
\ref{lemma:OinC} effectively is given by $s^{-1}$, to be used in the
last part of the proof only. More precisely
\begin{equation}
  \label{eq:30}
  (1-\epsilon C_6)s^{-1}\leq |\dot
  \gamma_\epsilon|/|\gamma_\epsilon|\leq (1+\epsilon C_6)s^{-1}.
\end{equation}

We calculate using \eqref{gradP}, \eqref{eq:21}, Lemma
\ref{lemma:perturbed-case} \ref{item:7}) and the notation
$\omega=\gamma_\epsilon/|\gamma_\epsilon|$
\begin{subequations}
\begin{align}
  \label{eq:36}
P_{\perp}\dot
    \gamma_\epsilon&=O(\sqrt\epsilon)\,|\dot \gamma_\epsilon|,\\
  \tfrac d{ds} P(\gamma_\epsilon)&=
  \ket*{
    \frac{P_{\perp}\dot
      \gamma_\epsilon}{|\gamma_\epsilon|}
  }\bra{\omega}
  +\ket{\omega}\bra*{ \frac{P_{\perp}\dot
    \gamma_\epsilon}{|\gamma_\epsilon|}}=O(\sqrt\epsilon)\,\tfrac{|\dot \gamma_\epsilon|}{|\gamma_\epsilon|};\label{eq:33}
\end{align}
\end{subequations}
here and below the notation $O(\sqrt\epsilon)$ is used for any
(matrix-valued) function
of $s\in [0,1]$ obeying $|O(\sqrt\epsilon)_{ij}|\leq \sqrt\epsilon C$ uniformly in
$s$, $x$, $G_\epsilon$ and $\gamma_\epsilon$. In fact we can above
choose $C=\sqrt{ 2C_1}$ and $C=2\sqrt{ 2C_1}$, respectively.

From \eqref{eq:33} we obtain
\begin{equation}
  \label{eq:37}
  \tfrac d{ds} \parb{P(\gamma_\epsilon)h}=P(\gamma_\epsilon)\dot
  h+O(\sqrt\epsilon)\,\tfrac{|\dot
    \gamma_\epsilon|}{|\gamma_\epsilon|}h.
\end{equation}
Below we stick to the notation $\dot h_{\perp}={(h_{\perp})}^{\cdot}$,
although by \eqref{eq:37} we could have choosen an alternative
interpretation.

Now, using \eqref{gradG}, \eqref{gradP}, \eqref{eq:36}, \eqref{eq:37}  and Lemma
\ref{lemma:perturbed-case} \ref{item:7}) the analogue of \eqref{eq:31}
reads
\begin{align}
\MoveEqLeft 4\int_0^1 \dot \gamma_\epsilon\nabla G_\epsilon\cdot h\dot hds\nonumber\\
&=4\int_0^1\big\{  h_\perp\cdot \dot h_\perp   -h_\perp G\dot h_\perp+
hO(\sqrt\epsilon)\dot h+hO(\sqrt\epsilon)
 h\tfrac{|\dot
  \gamma_\epsilon|}{|\gamma_\epsilon|}\big\} \frac{|\dot
  \gamma_\epsilon|}{|\gamma_\epsilon|}ds.  \label{eq:31bb}
\end{align}
Similarly the  analogue of \eqref{eq:32}
reads
\begin{equation}
\int_0^1 \dot \gamma_\epsilon (\nabla^2 G_\epsilon;h,h)\dot \gamma_\epsilon ds =
2\int_0^1 \big\{-|h_\perp|^2+h_\perp G h_\perp+hO(\sqrt\epsilon)h\big\}\frac{|\dot
  \gamma_\epsilon|^2}{|\gamma_\epsilon|^2}ds.\label{eq:32bb}
\end{equation}

Next we compute using \eqref{eq:3}, \eqref{gradG}, \eqref{gradP} and  Lemma
\ref{lemma:perturbed-case} \ref{item:7})
\begin{equation*}
  (|\dot
  \gamma_\epsilon|/|
  \gamma_\epsilon|)^\cdot=- (|\dot
  \gamma_\epsilon|/|
  \gamma_\epsilon|)^2\parb{1+O(\sqrt\epsilon)}.
\end{equation*}
Whence the  analogue of \eqref{eq:34}
reads
\begin{equation}\label{eq:34bb}
  4h_\perp\cdot \dot h_\perp |\dot
  \gamma_\epsilon|/|
  \gamma_\epsilon|-2|h_\perp|^2(|\dot
  \gamma_\epsilon|/|
  \gamma_\epsilon|)^2=2(
h^2_\perp|\dot
  \gamma_\epsilon|/|
  \gamma_\epsilon|)^\cdot+hO(\sqrt\epsilon)h(|\dot
  \gamma_\epsilon|/|
  \gamma_\epsilon|)^2.
\end{equation}

We insert \eqref{eq:31bb} and \eqref{eq:32bb} into
\eqref{positivity}, integrate by parts using \eqref{eq:34bb} and obtain the following partial
analogue of \eqref{eq:35}
\begin{align}
\MoveEqLeft[0.5] \langle \partial^2_\kappa E_\epsilon(x,\kappa) h,h \rangle\label{eq:35bb}\\
&=2\int_0^1\bigg\{\dot hG_\epsilon\dot h - 2 h_\perp G\dot h_\perp\frac{|\dot
  \gamma_\epsilon|}{|\gamma_\epsilon|}   +  h_\perp Gh_\perp\frac{|\dot
  \gamma_\epsilon|^2}{|\gamma_\epsilon|^2}+
hO(\sqrt\epsilon)\dot h\frac{|\dot
  \gamma_\epsilon|}{|\gamma_\epsilon|}  +hO(\sqrt\epsilon)h\frac{|\dot
  \gamma_\epsilon|^2}{|\gamma_\epsilon|^2}    \bigg\}ds.\nonumber
\end{align}

Finally we invoke
\eqref{eq:30} and \eqref{eq:11}, and obtain from \eqref{eq:35bb} using
again
\eqref{eq:37} the following analogue of \eqref{eq:35n}
\begin{equation}\label{eq:38}
 \langle \partial^2_\kappa E_\epsilon(x,\kappa) h,h \rangle
\geq 2\int_0^1 ( a-\sqrt \epsilon C)|\dot h|^2\geq  a\|h\|^2.
\end{equation}
 \end{proof}

 \begin{proof}[Proof of Theorem \ref{thm:main result}~\ref{item:1}).]
   Let $\Phi_\epsilon=\exp_{0,\epsilon}(1\cdot):T\R^d_0\to \R^d$
   denote the exponential map for the perturbed metric $G_\epsilon$
   (close to a given $G\in \vO$) at the point $0\in \R^d$ and
   evaluated at time one.  The positivity of the Hessian along any
   perturbed geodesic as guaranteed by Lemma
   \ref{lemma:2_perturbed-case} implies that $\Phi_\epsilon$ is a
   local diffeomorphism, cf. \cite[Theorem 2.16]{Ch} and \cite[Theorem
   14.1]{Mi}. By Lemma \ref{lemma:existence} $\Phi_\epsilon$ maps
   $T\R^d_0$ onto $\R^d$. By~\eqref{estminbb2} $\Phi_\epsilon$ is
   proper, whence $\Phi_\epsilon$ is one-to-one, cf. \cite[Theorem
   5.1.4]{Be}. We have verified the uniqueness property of Definition
   \ref{def:regmet} \ref{item: Condition 1}).

   The uniform positivity property of Definition \ref{def:regmet}
   \ref{item: Condition 2}) follows from Lemma
   \ref{lemma:2_perturbed-case} (with $c=a$ in
   \eqref{positivity}). Whence indeed $G_\epsilon\in \vU$ for
   $\epsilon\leq \epsilon_0$. 
 \end{proof}

\section{Proof of Theorem \ref{thm:main
    result}~\ref{item:2})} \label{Proof of Theorem thm:main result2}
We shall prove the bounds \eqref{eq:errestza}, \eqref{eq:errestzbB}
and \eqref{eq:errestz}. Let $G\in
  \vO$ be given. Using  the convention of Subsection
\ref{subsec:perturbed-case}  we  write $\tilde G=G_\epsilon$ for
perturbed metrics, and
we shall again
require $\|G_\epsilon-G\|_l\leq \epsilon\leq \epsilon_0$ for some
small $\epsilon_0>0$.

\myparagraph{Re \eqref{eq:errestza}.} By the variational definition
 \eqref{maxsol} and \eqref{estmin},
 $S^2(x)-|x|^2=O(\epsilon)|x|^2$
 uniformly in ${x\in \R^d\setminus{\{0\}}}$.
  Whence we obtain  uniformly in ${x\in \R^d\setminus{\{0\}}}$
\begin{equation}
  \label{eq:39}
  S(x)-|x|=( S^2(x)-|x|^2)/( S(x)+|x|)=O({\epsilon})|x|,
\end{equation} yielding \eqref{eq:errestza} by division with  $|x|$.

\myparagraph{Re \eqref{eq:errestzbB}.} We compute
\begin{equation}
  \label{eq:40}
  \partial _is(x)=\frac{ \partial _iS(x)}{|x|} -\frac{
    x_iS(x)}{|x|^3}.
\end{equation}
For the first term we use \eqref{eq:2}, \eqref{eq:errestza},
\eqref{estminbb1q} and
Lemma \ref{lemma:perturbed-case} \ref{item:7})
to write
\begin{equation}
  \label{eq:41}
 |x|\partial _iS(x)=x_i+O({\sqrt\epsilon})|x|.
\end{equation} As for second term we use  \eqref{eq:errestza} to
write
\begin{equation}
  \label{eq:42}
  \frac{
    x_iS(x)}{|x|}=x_i+O({\epsilon})|x|.
\end{equation}  Clearly  \eqref{eq:errestzbB} follows from
\eqref{eq:40}--\eqref{eq:42}.

\myparagraph{Re \eqref{eq:errestz}.}

\myparagraph{Step I.} We  prove the uniform
bound
\begin{equation}
    \label{eq:6ii}
    \|\partial_x^\alpha\kappa\|\leq \sqrt \epsilon C\mfor |\alpha|= 1.
  \end{equation} Note that \eqref{eq:6ii} without the factor $\sqrt
  \epsilon$ to the right follows from Proposition
  \ref{prop:soleikeq}. Due to Lemma \ref{lemma:2_perturbed-case} it suffices
  to bound the expression \eqref{eq:7} for $n=1$ as
\begin{equation}\label{eq:7ii}
  |\inp{\partial^\alpha_x\partial_\kappa
  E_\epsilon(x,\kappa),h}|\leq \sqrt \epsilon C\|h\| \mfor \alpha=e_j.
\end{equation} As in the proof of Proposition \ref{prop:soleikeq}  we
compute the $x$-derivative using  \eqref{graden} yielding four terms. Up to an
error of order $O(\epsilon)\|h\|$ we can replace $G_\epsilon\to G$,
 and whence
it suffices to show that  the sum of the following four expressions is  of
order $O(\sqrt\epsilon)\|h\|$:
\begin{align*}
  T_1&=\int_0^1 2e_j\nabla G\cdot h\dot\gamma_\epsilon\,ds,\\
T_2&=\int_0^1 2\dot h Ge_j\,ds,\\
T_3&=\int_0^1 s\dot\gamma_\epsilon \nabla \partial_jG\cdot h
\dot\gamma_\epsilon \,ds,\\
T_4&=\int_0^1 2s\dot h\partial_jG \dot\gamma_\epsilon \,ds.
\end{align*}
To do this we use   \eqref{eq:11}, \eqref{gradG}, \eqref{gradP}, \eqref{eq:21},
\eqref{eq:30}   and Lemma \ref{lemma:perturbed-case} \ref{item:7}) and
obtain
\begin{align*}
  T_1&=\int_0^1 2e_j (P_\perp -P_\perp GP_\perp)\tfrac hs\,ds+O(\sqrt\epsilon)\|h\|,\\
T_2&=\int_0^1 2e_j (P+P_\perp GP_\perp)\dot h\,ds,\\
T_3&=\int_0^1 2e_j ( P_\perp GP_\perp-P_\perp)\tfrac hs\,ds+O(\sqrt\epsilon)\|h\|,\\
T_4&=\int_0^1 2e_j ( P_\perp-P_\perp GP_\perp) \dot h\,ds+O(\sqrt\epsilon)\|h\|.
\end{align*} Clearly it follows that
\begin{align*}
  T_1+T_3&=O(\sqrt\epsilon)\|h\|,\\
T_2+T_4 &=\int_0^1 2e_j\dot h\,ds+O(\sqrt\epsilon)\|h\|=O(\sqrt\epsilon)\|h\|.
\end{align*} Whence we have proved \eqref{eq:7ii}.

\myparagraph{Step II.} We shall prove the uniform
bound
\begin{equation}
  \label{eq:43}
  |\partial_x^\alpha\dot \kappa(1)|=|\partial_x^\alpha\dot \gamma_\epsilon(1)-e_j|\leq  \sqrt\epsilon C \mfor \alpha=e_j.
\end{equation}

We claim that
\begin{equation}
  \label{eq:44}
  \partial_j\ddot \kappa(s)=s^{-1}F(s)\text{ where
  }\int_0^1|F(s)|^2\,ds\leq \epsilon C^2.
\end{equation}

 Note that \eqref{eq:43} follows from \eqref{eq:5}, \eqref{eq:6ii}  and
\eqref{eq:44}. The difficulty   here is not to show
that the quantity $s\partial_j\ddot \kappa(s)$ is square
integrable but rather to show that its
$L^2$-norm  is bounded by  $\sqrt\epsilon C$ as stated in
\eqref{eq:44}.
We shall show that
\begin{equation}
  \label{eq:44ii}
  \partial_j\ddot \kappa=G^{-1}P_{\perp} \nabla G\cdot  \dot
  \gamma_\epsilon
  P_{\perp}\parb{s^{-1}\partial_j\kappa-\partial_j\dot\kappa}+s^{-1}\tilde
  F \text{ where
  }\|\tilde F\|_{L^2}\leq \sqrt \epsilon C,
\end{equation} which combined with \eqref{eq:30} and \eqref{eq:6ii} supplies \eqref{eq:44}.

To prove \eqref{eq:44ii} we proceed as follows. Differentiating
\eqref{eq:3} with respect to $x_j$ and using \eqref{eq:30},
\eqref{eq:11},
\eqref{eq:6ii} and the facts  that
$\partial_j\gamma_\epsilon =se_j+\partial_j\kappa $ and $\|G_\epsilon-G\|_l\leq \epsilon$ we obtain
\begin{align}
\label{derkap}
G_\epsilon\partial_j\ddot\kappa ={}&\partial_j\dot\gamma_\epsilon\nabla G\cdot e_\bullet\dot\gamma_\epsilon
+2^{-1} \dot\gamma_\epsilon \partial_j\nabla G\cdot e_\bullet\dot\gamma_\epsilon-\partial_j\nabla G\cdot\dot\gamma_\epsilon\dot\gamma_\epsilon\\
 & -\nabla G\cdot\partial_j\dot\gamma_\epsilon \dot\gamma_\epsilon
 -\nabla G\cdot\dot\gamma_\epsilon \partial_j\dot\gamma_\epsilon- \partial_jG\ddot\kappa+O(\epsilon)/s\nonumber,
\end{align}
where here and henceforth $O(\epsilon^p)$ stands for a function with
$L^2$-norm bounded by~$\epsilon^p C$.

In the remaining of the proof of \eqref{eq:44ii} we will repeatedly
use  \eqref{gradG}, \eqref{gradP}, \eqref{eq:30}, \eqref{eq:36},
\eqref{eq:11}  and \eqref{eq:6ii}.

First we estimate $  \partial_jG\ddot\kappa$. To this end we note as above that
\begin{equation*}
 \partial_jG\ddot\kappa= \partial_jG G^{-1}(2^{-1}\dot\gamma_\epsilon\nabla G\cdot e_\bullet\dot\gamma_\epsilon-\nabla G\cdot\dot\gamma_\epsilon\dot\gamma_\epsilon)+O(\epsilon)/s.
\end{equation*}We compute
\begin{align*}
\partial_jG G^{-1} \dot\gamma_\epsilon\nabla G\cdot
e_\bullet\dot\gamma_\epsilon ={} & \nabla G\cdot\partial_j\gamma_\epsilon
G^{-1} [\dot\gamma_\epsilon \nabla P\cdot
e_\bullet\dot\gamma_\epsilon+\dot\gamma_\epsilon P_{\perp}\nabla
G\cdot e_\bullet P_{\perp}\dot\gamma_\epsilon \\ 
&+\dot\gamma_\epsilon \nabla P_{\perp}\cdot e_\bullet
GP_\perp\dot\gamma_\epsilon+\dot\gamma_\epsilon P_{\perp} G\nabla
P_{\perp}\cdot e_\bullet\dot\gamma_\epsilon]\\ 
={}& \nabla G\cdot\partial_j\gamma_\epsilon G^{-1} \dot\gamma_\epsilon \nabla P\cdot e_\bullet\dot\gamma_\epsilon+O(\sqrt\epsilon)/s\\
 ={}&O(\sqrt\epsilon)/s.
\end{align*}
Similarly we obtain $ \partial_jG G^{-1}\nabla G\cdot\dot\gamma_\epsilon\dot\gamma_\epsilon=O(\sqrt\epsilon)/s$, and therefore
\begin{equation}
\label{est0}
\partial_jG\ddot\kappa= O(\sqrt\epsilon)/s.
\end{equation}

 Now we estimate the first five terms on the right side of \eqref{derkap}.

\myparagraph{(i)} For the first term  we have
 \begin{align*}
\partial_j \dot\gamma_\epsilon\nabla G\cdot e_\bullet\dot\gamma_\epsilon&=
\partial_j \dot\gamma_\epsilon[\nabla P\cdot e_\bullet+P_\perp\nabla G\cdot e_\bullet P_\perp+\nabla P_\perp\cdot e_\bullet GP_\perp+ P_\perp G\nabla P_\perp\cdot e_\bullet  ]\dot\gamma_\epsilon\\
&=  \partial_j \dot\gamma_\epsilon \nabla P\cdot e_\bullet \dot\gamma_\epsilon+\partial_j \dot\gamma_\epsilon P_\perp G\nabla P_\perp\cdot e_\bullet \dot\gamma_\epsilon+O(\sqrt\epsilon)/s\\
&=(\langle\hat \gamma_\epsilon,
\dot\gamma_\epsilon\rangle/|\gamma_\epsilon|)(\langle \partial_j
\dot\gamma_\epsilon,P_\perp e_\bullet\rangle-\langle
GP_\perp \partial_j \dot\gamma_\epsilon, P_\perp e_\bullet\rangle)+
O(\sqrt\epsilon)/s.
\end{align*}
Thus,
\begin{equation}
\label{est1}
\partial_j \dot\gamma_\epsilon\nabla G\cdot e_\bullet\dot\gamma_\epsilon= (\langle\hat \gamma_\epsilon, \dot\gamma_\epsilon\rangle/|\gamma_\epsilon|)( I-P_\perp G)P_\perp \partial_j \dot\gamma_\epsilon+ O(\sqrt\epsilon)/s.
\end{equation}

\myparagraph{(ii)} To estimate the second term we  consider
  \begin{align*}
    \dot\gamma_\epsilon \partial_j\nabla G \cdot e_\bullet
    \dot\gamma_\epsilon={}&
    \dot\gamma_\epsilon \nabla[\nabla P\cdot e_\bullet+P_\perp\nabla G\cdot e_\bullet P_\perp\\
    &+\nabla P_\perp\cdot e_\bullet GP_\perp+P_\perp G\nabla P_\perp\cdot e_\bullet]\cdot \partial_j \gamma_\epsilon \dot\gamma_\epsilon\\
    ={}& \dot\gamma_\epsilon \nabla( \nabla P\cdot
    e_\bullet)\cdot \partial_j \gamma_\epsilon \dot\gamma_\epsilon
    +\dot\gamma_\epsilon \nabla( \nabla P_\perp\nabla \cdot e_\bullet G P_\perp)\cdot \partial_j \gamma_\epsilon \dot\gamma_\epsilon\\
    &+\dot\gamma_\epsilon\nabla(P_\perp G\nabla P_\perp\cdot
    e_\bullet)\cdot \partial_j \gamma_\epsilon \dot\gamma_\epsilon+
    O(\sqrt\epsilon)/s.
\end{align*}
It follows that
\begin{align}
\label{est3a}
\dot\gamma_\epsilon \partial_j\nabla G \cdot e_\bullet
\dot\gamma_\epsilon={}&\dot\gamma_\epsilon \nabla( \nabla P\cdot
e_\bullet)\cdot \partial_j \gamma_\epsilon \dot\gamma_\epsilon
+\dot\gamma_\epsilon  \nabla P\cdot e_\bullet G \nabla P\cdot  \partial_j \gamma_\epsilon \dot\gamma_\epsilon\nonumber\\
& + \dot\gamma_\epsilon \nabla P\cdot \partial_j \gamma_\epsilon G
\nabla P\cdot e_\bullet \dot\gamma_\epsilon+ O(\sqrt\epsilon)/s.
\end{align}
Using the fact that for any fixed vectors $h$ and $z$ we have $\nabla(\hat x)\cdot h=P_\perp h/|\gamma_\epsilon|$ and
\begin{equation*}
\nabla\left( \frac{P_\perp z}{|x|}\right)\cdot h=-\frac{1}{|x|^3}\left[\langle x,z\rangle P_\perp h+\langle P_\perp h,z\rangle x+\langle x,h\rangle P_\perp z\right],
\end{equation*} we find that
\begin{align*}
  \dot\gamma_\epsilon \nabla( \nabla P\cdot e_\bullet)\cdot \partial_j
  \gamma_\epsilon \dot\gamma_\epsilon&=
  \dot\gamma_\epsilon \nabla\left[ \Big| \frac{P_\perp
      e_\bullet}{|\gamma_\epsilon|} \Big\rangle\langle \hat
    \gamma_\epsilon|+|\hat \gamma_\epsilon
    \rangle\Big\langle\frac{P_\perp
      e_\bullet}{|\gamma_\epsilon|}\Big|\right]\cdot \partial_j
  \gamma_\epsilon \dot\gamma_\epsilon\\ 
  &=2 \langle \hat \gamma_\epsilon, \dot\gamma_\epsilon\rangle
  \left\langle \dot\gamma_\epsilon,\nabla\left( \frac{P_\perp
        e_\bullet}{|\gamma_\epsilon|}\right)\cdot \partial_j
    \gamma_\epsilon \right\rangle  +O(\sqrt\epsilon)/s\\ 
  &= -2\langle \hat \gamma_\epsilon, \dot\gamma_\epsilon\rangle
  \left\langle \dot\gamma_\epsilon, \frac{\langle P_\perp \partial_j
      \gamma_\epsilon, e_\bullet\rangle \gamma_\epsilon
    }{|\gamma_\epsilon|^3} \right\rangle +O(\sqrt\epsilon)/s.
\end{align*}
Hence
  \begin{equation*}
    \dot\gamma_\epsilon \nabla( \nabla P\cdot
    e_\bullet)\cdot \partial_j \gamma_\epsilon \dot\gamma_\epsilon =
    -2\frac{\langle
      \dot\gamma_\epsilon,\hat\gamma_\epsilon\rangle^2}{|\gamma_\epsilon|^2}P_\perp \partial_j
    \gamma_\epsilon +O(\sqrt\epsilon)/s.
\end{equation*}
Moreover, using twice \eqref{gradP} we obtain
\begin{align*}
 \dot\gamma_\epsilon \nabla P\cdot \partial_j\gamma_\epsilon G \nabla
 P \cdot e_\bullet \dot\gamma_\epsilon 
&= \dot\gamma_\epsilon \left[ \Big|\frac{P_\perp \partial_j
    \gamma_\epsilon}{|\gamma_\epsilon|}\Big\rangle   \langle
  \hat\gamma_\epsilon|+|\hat\gamma_\epsilon\rangle\Big\langle
  \frac{P_\perp \partial_j \gamma_\epsilon}{|\gamma_\epsilon|}\Big|
\right] G\nabla P\cdot e_\bullet  \dot\gamma_\epsilon\\ 
&= \langle\dot\gamma_\epsilon,\hat\gamma_\epsilon \rangle
\left\langle \frac{GP_\perp \partial_j
    \gamma_\epsilon}{|\gamma_\epsilon|},   \nabla P\cdot e_\bullet
  \dot\gamma_\epsilon\right\rangle+O(\sqrt\epsilon)/s\\ 
&= (\langle\dot\gamma_\epsilon,\hat\gamma_\epsilon
\rangle^2/|\gamma_\epsilon|^2 )P_\perp GP_\perp \partial_j
\gamma_\epsilon+O(\sqrt\epsilon)/s. 
\end{align*}
A similar calculation supplies
\begin{equation*}
\dot\gamma_\epsilon \nabla P\cdot e_\bullet G \nabla P\cdot\partial_j\gamma_\epsilon \dot\gamma_\epsilon= (\langle\dot\gamma_\epsilon,\hat\gamma_\epsilon \rangle^2/|\gamma_\epsilon|^2 )P_\perp GP_\perp \partial_j \gamma_\epsilon+O(\sqrt\epsilon)/s,
\end{equation*} and  putting the last three estimates together in \eqref{est3a} gives
\begin{equation}
\label{est2}
2^{-1}\dot\gamma_\epsilon \partial_j\nabla G \cdot e_\bullet \dot\gamma_\epsilon=(\langle\dot\gamma_\epsilon,\hat\gamma_\epsilon \rangle^2/|\gamma_\epsilon|^2 )( P_\perp G -I)P_\perp \partial_j \gamma_\epsilon
 +O(\sqrt\epsilon)/s
\end{equation}

\myparagraph{(iii)} Concerning the third term of \eqref{derkap} we have
 \begin{align*}
   \partial_j\nabla G\cdot\dot\gamma_\epsilon\dot\gamma_\epsilon
   ={} & \nabla(\nabla P\cdot\dot\gamma_\epsilon)\cdot \partial_j\gamma_\epsilon \dot\gamma_\epsilon+ \nabla(P_\perp\nabla G\cdot \dot\gamma_\epsilon P_\perp)\cdot\partial_j\gamma_\epsilon\dot\gamma_\epsilon\\
   &+\nabla(\nabla P_\perp\cdot\dot\gamma_\epsilon GP_\perp)\cdot\partial\gamma_\epsilon\dot\gamma_\epsilon+ \nabla(P_\perp G\nabla P_\perp \cdot\dot\gamma_\epsilon)\cdot\partial_j\gamma_\epsilon\dot\gamma_\epsilon\\
   ={}& \nabla(\nabla P\cdot\dot\gamma_\epsilon)\cdot \partial_j\gamma_\epsilon \dot\gamma_\epsilon+ P_\perp\nabla G\cdot \dot\gamma_\epsilon \nabla P_\perp \cdot\partial_j\gamma_\epsilon \dot\gamma_\epsilon \\
   &+P_\perp G\nabla(\nabla P_\perp\cdot
   \dot\gamma_\epsilon)\cdot \partial_j\gamma_\epsilon\dot\gamma_\epsilon
   +O(\sqrt\epsilon)/s.
\end{align*}
Analogous calculations to the ones leading to \eqref{est2} yield
\begin{align*}
  \nabla(\nabla
  P\cdot\dot\gamma_\epsilon)\cdot \partial_j\gamma_\epsilon
  \dot\gamma_\epsilon&=-(\langle
  \hat\gamma_\epsilon,\dot\gamma_\epsilon\rangle/|\gamma_\epsilon|)^2P_\perp\partial_j\gamma_\epsilon
  +O(\sqrt\epsilon)/s, 
\\
  P_\perp\nabla G\cdot \dot\gamma_\epsilon \nabla P_\perp
  \cdot\partial_j\gamma_\epsilon \dot\gamma_\epsilon &=-(\langle
  \hat\gamma_\epsilon,\dot\gamma_\epsilon\rangle/|\gamma_\epsilon|
  )P_\perp \nabla G\cdot \dot\gamma_\epsilon
  P_\perp \partial_j\gamma_\epsilon  +O(\sqrt\epsilon)/s, 
\\
\shortintertext{and}
P_\perp G\nabla(\nabla P_\perp\cdot
\dot\gamma_\epsilon)\cdot \partial_j\gamma_\epsilon\dot\gamma_\epsilon
&=(\langle
\hat\gamma_\epsilon,\dot\gamma_\epsilon\rangle/|\gamma_\epsilon|)^2
P_\perp  G P_\perp \partial_j\gamma_\epsilon  +O(\sqrt\epsilon)/s. 
\end{align*}
Therefore
\begin{align}
\label{est3}
\partial_j\nabla G\cdot\dot\gamma_\epsilon\dot\gamma_\epsilon
={}&-(\langle
\hat\gamma_\epsilon,\dot\gamma_\epsilon\rangle/|\gamma_\epsilon|)^2P_\perp\partial_j\gamma_\epsilon-(\langle
\hat\gamma_\epsilon,\dot\gamma_\epsilon\rangle/|\gamma_\epsilon|
)P_\perp \nabla G\cdot \dot\gamma_\epsilon
P_\perp \partial_j\gamma_\epsilon\nonumber \\ 
&+(\langle
\hat\gamma_\epsilon,\dot\gamma_\epsilon\rangle/|\gamma_\epsilon|)^2
P_\perp G P_\perp \partial_j\gamma_\epsilon +O(\sqrt\epsilon)/s,
\end{align}

\myparagraph{(iv)} To estimate the fourth term of \eqref{derkap} we proceed as follows
\begin{align}
\label{est4}
\nabla G\cdot\partial_j\dot\gamma_\epsilon\dot\gamma_\epsilon={}&(\nabla P\cdot\partial_j \dot\gamma_\epsilon+P_\perp\nabla G\cdot\partial_j \dot\gamma_\epsilon P_\perp  \nonumber\\
&+\nabla P_\perp\cdot\partial_j \dot\gamma_\epsilon GP_\perp+ P_\perp G\nabla P_\perp\cdot\partial_j \dot\gamma_\epsilon)\dot\gamma_\epsilon\nonumber\\
={}& \nabla P\cdot\partial_j \dot\gamma_\epsilon+ P_\perp G\nabla P_\perp\cdot\partial_j \dot\gamma_\epsilon\dot\gamma_\epsilon+O(\sqrt\epsilon)/s\nonumber\\
={}&(\langle\hat\gamma_\epsilon,\dot\gamma_\epsilon\rangle/|\gamma_\epsilon|)(I-P_\perp G)P_\perp\partial_j \dot\gamma_\epsilon+O(\sqrt\epsilon)/s.
\end{align}

\myparagraph{(v)} For the fifth term of \eqref{derkap} we have
\begin{align}
\label{est5}
\nabla G\cdot\dot\gamma_\epsilon\partial_j\dot\gamma_\epsilon={}& (\nabla P\cdot \dot\gamma_\epsilon+P_\perp\nabla G\cdot\dot\gamma_\epsilon P_\perp \nonumber\\
& +\nabla P_\perp\cdot\dot\gamma_\epsilon GP_\perp +P_\perp G\nabla P_\perp \cdot\dot\gamma_\epsilon )\partial_j \dot\gamma_\epsilon\nonumber\\
={}&P_\perp\nabla G\cdot\dot\gamma_\epsilon P_\perp \partial_j
\dot\gamma_\epsilon +O(\sqrt\epsilon)/s.
\end{align}
Substituting \eqref{est0}, \eqref{est1}, \eqref{est2}, \eqref{est3}, \eqref{est4} and \eqref{est5} in \eqref{derkap} we find that
\begin{equation*}
G_\epsilon\partial_j\ddot\kappa=P_\perp \nabla G \cdot\dot\gamma_\epsilon P_\perp \left((\langle\hat\gamma_\epsilon,\dot\gamma_\epsilon\rangle/|\gamma_\epsilon|)\partial_j \gamma_\epsilon- \partial_j \dot\gamma_\epsilon\right) +O(\sqrt\epsilon)/s,
\end{equation*} and using now \eqref{eq:30} and Lemma
\ref{lemma:perturbed-case} \ref{item:7}) we arrive at
\eqref{eq:44ii} due to cancellation.

\myparagraph{Step III.} We shall complete the proof of \eqref{eq:errestz}
by differentiating \eqref{eq:40} to obtain a representation of the
second order derivatives and then using \eqref{eq:errestza},
\eqref{eq:errestzbB}, \eqref{eq:2}
and \eqref{eq:43}. We have
\begin{equation}
  \label{eq:40s}
 \partial _{ij}s(x)=\frac{ \partial _{ij}S}{|x|} -\frac{
    x_j\partial _{i}S+x_i\partial _{j}S}{|x|^3}-\frac{
    \delta_{ij}S}{|x|^3} + 3\frac{ x_ix_jS}{|x|^5}.
\end{equation}
Using \eqref{eq:errestza} and  \eqref{eq:errestzbB} it follows from
\eqref{eq:40s} that
\begin{equation}
  \label{eq:40ii}
 \partial _{ij}s(x)-\frac{ \partial _{ij}S}{|x|} =\frac{
    \hat x_i\hat x_j-\delta_{ij}+O(\sqrt{\epsilon})}{|x|^2}.
\end{equation}

As for the first term on the right hand side of
\eqref{eq:40s}  we use \eqref{eq:2}
and \eqref{eq:43} to write
\begin{equation}
  \label{eq:46}
  \frac{ \partial _{j}\nabla S}{|x|}- \frac{
    S^{-1}G\parb{e_j+O(\sqrt{\epsilon})}}{|x|}=\frac{ (\partial_j S^{-1})G\dot
    \gamma_\epsilon(1)}{|x|}+\frac{ S^{-1}\nabla G\cdot e_j\dot \gamma_\epsilon(1)}{|x|}.
\end{equation} The $i$th component of the right hand side of \eqref{eq:46} is of the form
\begin{equation*}
  -\frac{
    \hat x_i\hat x_j}{|x|^2}+|x|^{-2}\parb{(P_\perp -P_\perp
    GP_\perp)e_j }_i+\frac{O(\sqrt{\epsilon})}{|x|^2}.
\end{equation*} Obviously
\begin{equation*}
  \frac{
    \parb{S^{-1}Ge_j}_i}{|x|}=\frac{\parb{
    Ge_j}_i}{|x|^2}+\frac{O(\sqrt{\epsilon})}{|x|^2}.
\end{equation*} By adding these expressions we obtain that
\begin{equation}\label{eq:47}
  \frac{ \partial _{ij}S}{|x|} =-\frac{
    \hat x_i\hat x_j-\delta_{ij}+O(\sqrt{\epsilon})}{|x|^2}.
\end{equation} We combine \eqref{eq:40ii} and \eqref{eq:47} and
conclude \eqref{eq:errestz}.\qed

\begin{remark}
  \label{remark:proof-theor-refthm:m}
  An analogue of \eqref{eq:6ii} was proved by a different technique in
  \cite{Ba} (see~\cite[Lemma 3]{Ba}) which does not work in the
  present context. Note that in \cite{Ba} also an equation like
  \eqref{eq:44ii} is used, and in fact it is used to prove the
  analogue of~\eqref{eq:6ii}. Howewer in our case there is no
  smallness of the factor $G^{-1}P_{\perp} \nabla G\cdot \dot
  \gamma_\epsilon P_{\perp}$ and consequently the technique in
  \cite{Ba} is not applicable. Note also that \cite[Lemma 4]{Ba} is an
  analogue of \eqref{eq:43}.
\end{remark}

\section{Proof of Theorem \ref{thm:main-result2}} \label{Proof of
  Theorem main result2}
In this section we need the Sobolev spaces $\vH^p:=W_0^{1,p}(0,1)^d$, $1<p<\infty$,
consisting of absolutely continuous functions $h:[0,1]\to \R^d$
vanishing at the endpoints and having $\dot h\in L^p(0,1)^d=L^p(]0,1[,\R^d)$
(throughout this section we use the notation $L^p$ for this
vector-valued $L^p$ space). The
space  $\vH^p$ is equipped with the  norm
\begin{equation}
  \label{eq:54}
  \|h\|_{\vH^p}=\|\dot h\|_p=\parbb{\int_0^1|\dot h(s)|^p ds}^{1/p}.
\end{equation}
The first goal is to find a substitute for the positivity bound
\eqref{positivitycc} used in the proof of Theorem \ref{thm:main
  result}. We shall use the observation, cf. the proof of Lemma
\ref{lemma:OinC},  that
\begin{equation}\label{eq:35nb}
\langle \partial^2_\kappa E(x,0) h_1,h_2 \rangle = \int_0^1  2s^2\big\{ \big( s^{-1}h_1(s)\big)^\cdot
G(sx)     \big( s^{-1} h_2(s)\big)^\cdot \big\}ds.
\end{equation}
Motivated by \eqref{eq:35nb}  we develop  some functional
analysis which then will  be  applied
to an extension/modification for perturbed
geodesics.
\subsection{Hardy type bounds and  duality
  theory} \label{Hardy spaces and their dual}
\begin{lemma}
  \label{lemma:hardy-type-bounds} For all $p\in]1,\infty[$ there
  exists $A_p>0$ such that
  \begin{equation}
    \label{eq:55}
    \|h(\cdot)/s\|_{p}=\parbb{\int_0^1| h(s)/s|^p ds}^{1/p}\leq A_p
    \|h\|_{\vH^p}\mforall h\in \vH^p.
  \end{equation}
\end{lemma}
  \begin{proof} We refer to \cite[Theorem 327]{HLP}; the bound is
    valid for $ A_p=p/(p-1)$.
\end{proof}

\begin{lemma}
  \label{lemma:hardy-type-bounds2} For all $p\in]1,\infty[$ there
  exists $B_p>0$ such that
  \begin{equation}
    \label{eq:55b}
    \|h\|_{\vH^p}\leq B_p\|\dot h(\cdot) -h(\cdot)/s\|_{p}\mforall h\in \vH^p.
  \end{equation}
\end{lemma}
\begin{proof} Consider the linear  maps
  \begin{subequations}
  \begin{align}
\label{eq:56}
    L^p(0,1)^d\ni f&\to Sf;\;(Sf)(t)=t^{-1}\int_0^t f(s)ds.
      \\\label{eq:56b}
    L^p(0,1)^d\ni f&\to Tf;\;(Tf)(t)=-t\int_t^1 s^{-1}f(s)ds.
    \end{align}
  \end{subequations}

Note that $S$ is bounded on
$L^p$ with $\|S\|_{\vB(L^p)}\leq A_p$,  cf.  Lemma
\ref{lemma:hardy-type-bounds} and its proof.
  Next note
  \begin{equation}\label{eq:60}
    \tfrac {d}{dt}(Tf)(t)=t^{-1}(Tf)(t)+f(t),
  \end{equation}
   and recall the standard fact
\begin{equation}
  \label{eq:63}
  L^p(I,\mathcal G)=(L^q(I,\vG)^* \text{ if }q^{-1}=1-p^{-1};
\end{equation}
in our
case $I=]0,1[$ and $\vG$ is the Hilbert space $\R^d$. In terms of
\eqref{eq:63}  we can rewrite
\eqref{eq:60} as
\begin{equation*}
    \dot T:=\tfrac
{d}{dt}T=-S^*+I,
\end{equation*}  and since $S\in \vB(L^q)$ we conclude that $\dot T \in \vB(L^p)$.

Finally noting that for any $h\in \vH^p$ we can write $h= T f$ where $ f(s)=\dot
h(s)-h(s)/s$ we obtain  \eqref{eq:55b} with $B_p=A_q+1$.
\end{proof}
\begin{remark}
  \label{remark:hardy-type-bounds} By the estimate (due to the
  H{\"o}lder inequality)
\begin{equation*}
 \big |\int_t^1 s^{-1}f(s)ds \big |\leq t^{-1/p}(p-1)^{1-1/p}\|f\|_{p},
\end{equation*} it follows from the proof of Lemma
\ref{lemma:hardy-type-bounds2} that  the operator $T$ of
\eqref{eq:56b} is in $\vB(L^p,\vH^p)$ with norm $\|\dot
T\|_{\vB(L^p)}\leq B_p$. In fact (using here also Lemma \ref{lemma:hardy-type-bounds}) $T:L^p
\rightarrow\vH^p$ is a  linear homeomorphism.
\end{remark}
\begin{lemma}
  \label{lemma:hardy-type-bounds3} Suppose $G$ is a given continuous
  function, $[0,1]\to G(s)\in
  \vS_d(\R)$, using here notation of Subsection \ref{Conditions and
    main results}. Suppose there are constants $a,b>0$ such that
\begin{equation}
   \label{ellcond3}
   a|y|^2\leq y G(s)y\leq b|y|^2,\qquad \quad y\in\R^d\mand s\in [0,1].
 \end{equation} Introduce for $h\in \vH^p$ and $g\in \vH^q$,
 $q^{-1}+p^{-1}=1$, the pairing
 \begin{subequations}
   \begin{equation}
   \label{eq:59}
   [h,g]=\int_0^1  2s^2\big\{ \big( s^{-1}h(s)\big)^\cdot
G(s)     \big( s^{-1} g(s)\big)^\cdot \big\}ds,
 \end{equation} and the associated quantity
 \begin{equation}
   \label{eq:61}
   \|h\|_{\vH^p,G}=\sup_{\|g\|_{\vH^q}\leq 1}|[h,g]|.
 \end{equation}
 \end{subequations}

This quantity  $\|\cdot \|_{\vH^p,G}$ is a norm on $\vH^p$ and
\begin{gather}
  \label{eq:62}
  C_1^{-1}\|h\|_{\vH^p,G}\leq
  \|h\|_{\vH^p}\leq C_2\|h\|_{\vH^p,G};\\
C_1=2b(1+A_q)(1+A_p),\;
C_2=(2a)^{-1}B_qB_p.\nonumber
\end{gather} Here $C_1$ and $C_2$  are given in terms of the constants of Lemmas
\ref{lemma:hardy-type-bounds}-\ref{lemma:hardy-type-bounds2} and
\eqref{ellcond3}.

Finally $\parb{\vH^p}^*=\vH^q$ in the sense given by the pairing
\eqref{eq:59}. This means that for all $l\in \parb{\vH^p}^*$ there
exists a unique $g\in \vH^q$ such that
\begin{equation}
  \label{eq:65s}
  l(h)= [h,g]\mforall h\in
\vH^p,
\end{equation}
 and vica versa any $g\in \vH^q$  defines an element
 $l\in \parb{\vH^p}^*$ by \eqref{eq:65s}. Moreover the identification is a linear homeomorphism.
\end{lemma}
\begin{proof}
  Notice that indeed the pairing  \eqref{eq:59} is well-defined
  due to  the H\"older and Minkowski inequalities  and Lemma
  \ref{lemma:hardy-type-bounds}. Using these bounds we also get the
  first estimate of \eqref{eq:62}.

To prove the second  estimate of \eqref{eq:62} we first use Lemma
\ref{lemma:hardy-type-bounds2} and \eqref{eq:63} to estimate
\begin{align*}
  \|h\|_{\vH^p}&\leq B_p\sup_{\|2Gf\|_{L^q}\leq 1}\Big |\int_0^1  s \big( s^{-1}h(s)\big)^\cdot
2G(s)f(s)ds \Big |\\&\leq B_p\sup_{\|f\|_{L^q}\leq (2a)^{-1}}\Big |\int_0^1  s \big( s^{-1}h(s)\big)^\cdot
2G(s)f(s)ds \Big |.
\end{align*} Next we introduce  for any $f\in L^q$ the function $g
  =Tf$ which according to Remark \ref{remark:hardy-type-bounds} is an
  element of $\vH^q$. We obtain
\begin{align*}
  \|h\|_{\vH^p}&\leq B_p\sup_{\|f\|_{L^q}\leq (2a)^{-1},\, g=Tf}\Big |\int_0^1  s \big( s^{-1}h(s)\big)^\cdot
2G(s)s \big( s^{-1}g(s)\big)^\cdot ds \Big |\\
&\leq B_p\sup_{\|g\|_{\vH^q}\leq (2a)^{-1}B_q}\Big |\int_0^1  s \big( s^{-1}h(s)\big)^\cdot
2G(s)s \big( s^{-1}g(s)\big)^\cdot ds \Big |\\
&=B_p(2a)^{-1}B_q\,\|h\|_{\vH^p,G}.
\end{align*}

We have proved \eqref{eq:62}.

The identification  asserted next follows similarly from  \eqref{eq:63}
and Remark \ref{remark:hardy-type-bounds}. The bi-continuity is a
consequence of \eqref{eq:62}.
\end{proof}

\begin{cor}
  \label{cor:hardy-type-bounds} Let $G$ be given as in Lemma
  \ref{lemma:hardy-type-bounds3},  and define $P\in \vB(\vH)$,
  $\vH=\vH^2$, by
  \begin{equation*}
    \inp{Ph_1, h_2}=[h_1, h_2];\;h_1, h_2\in \vH.
  \end{equation*} Let $C_1$
  and  $C_2$ be the constants in
  \eqref{eq:62}.

For all
  $p\in [2,\infty[$ the operator
  $P\in\vB\parb{\vH^p}$ and it obeys
  \begin{subequations}
   \begin{equation}
    \label{eq:64}
    \|P\|_{\vB(\vH^p)}\leq 2C_1.
  \end{equation} In particular (with $q\in ]1,2]$ being
  the conjugate exponent)
  \begin{equation}
    \label{eq:66}
    |\inp{Ph,g}|\leq 2C_1\|h\|_{\vH^p}\,\|g\|_{\vH^q}\mforall
    h\in \vH^p\mand g\in \vH^2.
  \end{equation}
\end{subequations}
\begin{subequations}
Conversely, if
  \begin{equation}
    \label{eq:68}
    h\in \vH^p \mand |\inp{Ph,g}|\leq c\|g\|_{\vH^q}\mforall
   g\in \vH^2,
  \end{equation} then
  \begin{equation}
    \label{eq:67}
   \|h\|_{\vH^p}\leq cC_2.
  \end{equation}
  \end{subequations}

\end{cor}
\begin{proof}
  We
  show that $Ph\in  \vH^p$ for  all $h\in  \vH^p$ if  $p\in
  [2,\infty[$: Fix any such $p$ and such  $h$, and let $q\in ]1,2]$ be
  the conjugate exponent. Then    $Ph\in  \vH^p$ if the
  expression
  \begin{equation*}
    \|Ph\|_{\vH^p}=\sup_{\|f\|_{L^q}\leq 1,f\in L^2}\big |\int_0^1   (Ph)^\cdot(s)
f(s)ds \big | <\infty.
  \end{equation*} We introduce the map
    \begin{equation}
      \label{eq:56f}
    L^q\ni f\to g=Rf;\;(Rf)(t)=\int_0^t \parbb{f(s)-\int_0^1f(l)dl}ds.
    \end{equation} Clearly $R$ maps into $\vH^q$ and in fact (due to
    the H\"older inequality)
    \begin{equation*}
      \|R\|_{\vB(L^q,\vH^q)}\leq 2.
    \end{equation*} Whence
\begin{align*}
  \|Ph\|_{\vH^p}&=\sup_{\|f\|_{L^q}\leq 1, f\in L^2}\Big |\int_0^1
  (Ph)^\cdot(s)
  (Rf)^\cdot(s)ds \Big |\\
  &\leq\sup_{\|g\|_{\vH^q}\leq 2,g\in {\vH^2}}\Big |\int_0^1
  (Ph)^\cdot(s)
  \dot g(s)ds \Big |\\
  &\leq\sup_{\|g\|_{\vH^q}\leq 2,g\in {\vH^2}}|[h,g]|\\
  &\leq 2C_1\|h\|_{\vH^p}.
  \end{align*}
So indeed $Ph\in  \vH^p$  and the estimates show \eqref{eq:64}.

Clearly \eqref{eq:66} follows from  the
H\"older inequality and  \eqref{eq:64}.

As for \eqref{eq:67} we note that by a density
argument and \eqref{eq:68}
\begin{equation*}
  \|h\|_{\vH^p,G}=\sup_{\|g\|_{\vH^q}\leq 1,g\in \vH^2}|[h,g]|=\sup_{\|g\|_{\vH^q}\leq 1,g\in \vH^2}|\inp{Ph,g}|\leq c.
\end{equation*} Whence we can conclude by \eqref{eq:62}.
\end{proof}
\begin{remark*} In our application, see Subsection \ref{Bounds of
    partial alpha kappa},  $G(s)$ will be a composition of a
  metric in $\vO$ and a perturbed geodesic  (as in  \eqref{eq:35nb} for
  the unperturbed case). We shall conclude
  \eqref{eq:67} for functions in question by verifying \eqref{eq:68}.

 For completeness of
  presentation let us note the following additional property of  the
  operator $P$ (shown to be in
  $\vB\parb{\vH^p}$):  For all $p\in [2,\infty[$ in fact $P$
  is a linear homeomorphism on $\vH^p$. Outline of  a
proof: We verify the condition \eqref{eq:68} using the H\"older inequality and
deduce from  the conclusion  \eqref{eq:67} that
\begin{equation*}
  \|h\|_{\vH^p}\leq C_2\|Ph\|_{\vH^p}\mforall h\in \vH^p.
\end{equation*} Whence $P:\vH^p \rightarrow \vH^p$ is injective
and has dense range. It remains to show that its range is dense. For
that we note that
\begin{equation*}
  \inp{h,g}=\int_0^1 \dot h\dot g \,d s,\; h\in\vH^p \mand g\in\vH^q,
\end{equation*} defines  a pairing giving another identification
$\parb{\vH^p}^*=\vH^q$ (use the first step of the proof of
Corollary \ref{cor:hardy-type-bounds} and  Hahn-Banach
theory). Whence if the range is not dense we can pick $0\neq
g\in\vH^q$ such that $\inp{Ph,g}=0$ for all $h\in\vH^p$,
 violating that $\inp{Ph,g}=[h,g]$ and the last part of
Lemma \ref{lemma:hardy-type-bounds3}.
\end{remark*}

\subsection{Bounds of $\partial^\alpha\kappa$ in
  $\vH^p$, $p\in[2,\infty[$} \label{Bounds of partial alpha kappa}
We shall improve on Proposition \ref{prop:soleikeq} \ref{item:5a}) in
the case of a metric $\tilde G\in\vM$ close to a given $G\in\vO$ (as
in Theorem \ref{thm:main result} \ref{item:2})). We shall use the
convention used in Subsection \ref{subsec:perturbed-case} and write in
terms of an ``order parameter''
$\tilde G=G_\epsilon$ if $\|G_\epsilon-G\|_l\leq \epsilon$.  Following the proof
of Proposition \ref{prop:soleikeq} \ref{item:5a}) we can for
$1\leq |\alpha|\leq l-1$ represent the
quantity
\begin{equation*}
  -\inp{\partial^2_\kappa E_\epsilon(x,\kappa)\partial^\alpha\kappa,h}
\end{equation*} as a sum of terms involving derivatives
$\partial^\beta\kappa$ where $|\beta|\leq |\alpha|-1$. Using a new
induction hypothesis we shall  estimate these terms
individually improving the corresponding bounds of the proof of  Proposition \ref{prop:soleikeq}
\ref{item:5a}). On the other hand we can write
\begin{subequations}
 \begin{align}\label{eq:65}
  \partial^2_\kappa
    E_\epsilon(x,\kappa)\phantom{:}={}& P+R;\\
\inp{Ph_1,h_2}=[h_1,h_2]:={}&\int_0^1  2s^2\big\{ \big( s^{-1}h_1(s)\big)^\cdot
G(\gamma_\epsilon(s))     \big( s^{-1} h_2(s)\big)^\cdot \big\}ds.\label{eq:69}
\end{align}

\end{subequations}

The second term  in \eqref{eq:65} can be estimated by the  following
\begin{lemma}
  \label{lemma:bounds-part-vhp} Let $p\in[2,\infty[$ and $q$ be the
  conjugate exponent. There exists $\epsilon_0>0$ and $C>0$ such
  that for  all $x\in\R^d$ and for $\|G_\epsilon-G\|_l\leq \epsilon\leq \epsilon_0$
  \begin{equation}
    \label{eq:70}
    |\inp{Rh_1,h_2}|\leq \sqrt\epsilon C\|h_1\|_{\vH^p}\,\|h_2\|_{\vH^q}\mforall
    h_1\in\vH_p\mand  h_2\in\vH_2.
  \end{equation}
\end{lemma}
\begin{proof} The result follows  from \eqref{eq:35bb}  using
  \eqref{eq:30}, \eqref{eq:37},  the H\"older
  inequality and Lemma \ref{lemma:hardy-type-bounds}.
  \end{proof}
  \begin{lemma}
    \label{lemma:bounds-part-vhp2} Let $p$,  $q$, $\epsilon_0>0$ and $C>0$  be given as in
    lemma \ref{lemma:bounds-part-vhp}. Suppose $1\leq  |\alpha|\leq l-1$  and
    that for some constant
    $C_\alpha>0$ independent of $x\in\R^d$ and $\epsilon\in [0, \epsilon_0]$
    \begin{equation}
      \label{eq:71}
      |\inp{\partial^2_\kappa
    E_\epsilon(x,\kappa)\partial^\alpha\kappa,h}|\leq C_\alpha\inp{x}^{1-|\alpha|}\|h\|_{\vH^q}\mforall
   h\in\vH^2.
    \end{equation} Let $a,b>0$ be the constants from \eqref{ellcond}
    determined by the metric $G$ and let $C_2$ be the constant from
    \eqref{eq:62} (fixed in terms of $a$ and $p$). Then
    \begin{equation}
      \label{eq:72}
      \parb{1-\sqrt\epsilon
        CC_2}\,\|\partial^\alpha\kappa\|_{\vH^p}\leq C_\alpha C_2\inp{x}^{1-|\alpha|}.
    \end{equation}
  \end{lemma}
  \begin{proof}
    Clearly \eqref{ellcond3} holds for the  example
    $G(\gamma_\epsilon(\cdot))$ used in  \eqref{eq:69}. From
    \eqref{eq:65}, \eqref{eq:69}
    and Lemma \ref{lemma:bounds-part-vhp} we deduce \eqref{eq:68} with
    \begin{equation*}
      c= C_\alpha\inp{x}^{1-|\alpha|}+\sqrt\epsilon
        C\|\partial^\alpha\kappa\|_{\vH^p},
    \end{equation*} and we obtain from \eqref{eq:67} that
\begin{equation*}
      \|\partial^\alpha\kappa\|_{\vH^p}\leq \parb{ C_\alpha\inp{x}^{1-|\alpha|}+\sqrt\epsilon
        C\|\partial^\alpha\kappa\|_{\vH^p}}C_2
    \end{equation*} yielding \eqref{eq:72} by  subtraction.
  \end{proof}
  \begin{prop}
    \label{prop:bounds-part-vhp}  Let $p\in[2,\infty[$. There exist $\epsilon_0>0$ and $C_p>0$ such
  that for  $x\in\R^d$, $\|G_\epsilon-G\|_l\leq \epsilon\leq \epsilon_0$ and
  $1\leq |\alpha|\leq l-1$
\begin{equation}
      \label{eq:72b}
      \|\partial^\alpha\kappa\|_{\vH^p}\leq C_p\inp{x}^{1-|\alpha|}.
    \end{equation}
  \end{prop}
  \begin{proof} Using \eqref{eq:72} for all sufficiently small
    $\epsilon$ we only need to demonstrate \eqref{eq:71}. So let us
    verify \eqref{eq:71} for the case $|\alpha|=1$: We mimic Step I in
    the proof of Theorem~\ref{thm:main result}\,\ref{item:2}) using the
    H\"older inequality and Lemma~\ref{lemma:hardy-type-bounds} (to
    replace the Cauchy Schwarz and Hardy inequalities,
    respectively). This yields \eqref{eq:71} for $|\alpha|=1$, in fact
    with any extra factor $\sqrt \epsilon$. We have shown the
    proposition for the case $l=2$.

    Next let us suppose $l\geq 3$ and that we know the bounds
    \eqref{eq:72b} for all $p\in[2,\infty[$ and for $|\alpha|\in[1,
    n-1]$ where $l-1\geq n\geq 2$ (notice that $\epsilon_0>0$ may
    depend on~$p$). Let $p\in[2,\infty[$ and $\alpha$ with $|\alpha|=
    n$ be given. The proof is complete if we can show the existence of
    $\epsilon_0>0$ such that \eqref{eq:71} holds for the given $p$ and
    $\alpha$ provided $\|G_\epsilon-G\|_l\leq \epsilon\leq
    \epsilon_0$. For that we need to modify the proof of
    Proposition~\ref{prop:soleikeq}\,\ref{item:5a}). Again we can
    assume that $|x|\geq 1$.  We shall use the induction hypothesis
    for
\begin{equation}
  \label{eq:73}
  p\rightarrow \tilde p:=4np.
\end{equation} This particular choice fixes some
$\epsilon_0>0$ that we claim indeed works for  \eqref{eq:71} (it is
not claimed to be an ``optimal'' choice of $\tilde p$ although growth in $n$ is indispensable). The  Hardy inequality \eqref{eq:11} is replaced by Lemma
\ref{lemma:hardy-type-bounds} and \eqref{eq:12} by
\begin{equation}
  \label{eq:12tt}
  |\tilde h(s)|\leq s^{1-1/{\tilde p}} \|\tilde h\|_{\vH^{\tilde p}};
\end{equation} here we shall use the $\tilde p$ of \eqref{eq:73}. We shall also need the generalized H\"older
  inequality
  \begin{equation}
    \label{eq:74}
    \int_0^1|f_1(s)|\cdots |f_m(s)|\,ds \leq \|f_1\|_{p_1}\cdots
    \|f_m\|_{p_m}\mfor 1/{p_1}+ 1/{p_m}\leq 1.
  \end{equation}
Finally we shall use the following special case of \eqref{eq:9}
\begin{equation}
  \label{eq:9pp}
  |\partial^\eta g_{ij}|\leq C
  |sx|^{-|\eta|};
\end{equation} here and henceforth we omit the subscript $\epsilon$
(obviously $C$ is independent of $\epsilon$). Note that
\eqref{eq:9pp}  is the best possible bound at infinity. The main
issue  compared to the proof of Proposition
\ref{prop:soleikeq} \ref{item:5a}) is that the improved pointwise
bound \eqref{eq:12tt}, used to $\tilde h=\partial^{\beta_j}\kappa$,
compensates  for the worse singularity at $s=0$ when using  \eqref{eq:9pp}.

Now let us look at some details: There are  cases A), B) and C)
defined as in the proof of Proposition
\ref{prop:soleikeq} \ref{item:5a}) (and we
treat only  $k\geq 1$).

For the case   A) we need  \eqref{eq:9pp} with $|\eta|=|\zeta|+k-1$, and
we
consider subcases  Ai)--Aiii) defined as before.

\myparagraph{Case Ai):}  $\partial^\eta
G(\gamma(s))=s^{-|\zeta|}\partial^\zeta_x\partial_{\kappa(s)}^{\omega}
G(sx+\kappa(s));\;|\omega|=k-1$. Suppose first that $k\geq 2$ and
$i,j\leq k$. Then we shall use \eqref{eq:74} with $m=4$,
$p_1<\tilde{p}/n$, $p_2 =p_3=\tilde{p}$ and $p_4=q$, $f_1(s)=s^{-n/\tilde{p}}$, $f_2=
\dot h_i$, $f_3=
\dot h_j$ and $f_4(s)=h(s)/s$. Note that this is legitimate if
$\tilde{p}/n-p_1>0$ is sufficiently  small; it  yields the
extra factor $s^{1+n/\tilde{p}}$ appearing  below.
Upon
using  the pointwise bound \eqref{eq:12tt} for the remaining $k-2$
factors of components of $h_{\bullet}$'s, \eqref{eq:9pp}  and the fact
(since $k\leq n$) that
\begin{equation*}
  s^{-|\eta|}s^{|\zeta|}s^{(k-2)\parb{1-1/{\tilde p}}}s^{1+n/\tilde p}\leq 1,
\end{equation*} we obtain
the bound
\begin{equation}
  \label{eq:13b}
  |(F_{\zeta,k};h_1,\dots, h_{k+1})|\leq
  C
  |x|^{-|\eta|}\prod_{m=1}^{k} \|h_m\||_{\vH^{\tilde p}}\,\|h\|_{\vH^q}.
\end{equation} By the induction hypothesis
\begin{equation*}
  \prod_{m=1}^{k} \|h_m\|_{\vH^{\tilde p}}\leq C\inp{x}^{k-\sum|\beta_m|}=C\inp{x}^{k-(n-|\zeta|)},
\end{equation*} which together with \eqref{eq:13b} yields
\begin{equation}
  \label{eq:14b}
  |(F_{\zeta,k};h_1,\dots, h_{k+1})|\leq
  C
  \inp{x}^{-|\eta|}\inp{x}^{k-(n-|\zeta|)}\|h\|_{\vH^q}= C\inp{x}^{1-n}\|h\|_{\vH^q}.
\end{equation}

Suppose next that $j=k+1$. Then we apply \eqref{eq:74} with $m=3$,
$p_1<\tilde{p}/n$, $p_2 =\tilde{p}$ and $p_3=q$, $f_1(s)=s^{-n/\tilde{p}}$, $f_2=
\dot h_i$ and $f_3=\dot h$  yielding an extra factor $s^{n/\tilde{p}}$. Upon
using  the pointwise bound \eqref{eq:12tt} for the remaining $k-1$
factors of components of  $h_{\bullet}$'s, \eqref{eq:9pp}  and the fact that
\begin{equation*}
  s^{-|\eta|}s^{|\zeta|}s^{(k-1)\parub{1-1/{\tilde p}}}s^{n/\tilde p}\leq 1,
\end{equation*} we obtain again \eqref{eq:13b} and whence \eqref{eq:14b}.

    \myparagraph{Case Aii):} $\partial^\eta
G(\gamma(s))=s^{-|\zeta_1|}\partial^{\zeta_1}_x\partial_{\kappa(s)}^{\omega}
G(sx+\kappa(s))$; $|\zeta_1|=|\zeta|-1 ,\;|\omega|=k$. Suppose first
that
$j\leq k$.
Then we apply \eqref{eq:74} with $m=3$,
$p_1<\tilde{p}/n$, $p_2 =\tilde{p}$ and $p_3=q$, $f_1(s)=s^{-n/\tilde{p}}$, $f_2=
\dot h_j$ and $f_3(s)= h(s)/s$ yielding an extra factor
$s^{1+n/\tilde{p}}$.
Upon
using  \eqref{eq:12tt} for the remaining $k-1$
factors of components of $h_{\bullet}$'s, \eqref{eq:9pp}  and the fact that
\begin{equation*}
  s^{-|\eta|}s^{|\zeta_1|}s^{(k-1)\parub{1-1/{\tilde p}}}s^{1+n/\tilde p}\leq 1,
\end{equation*} we obtain again \eqref{eq:13b} and whence \eqref{eq:14b}.

Suppose next that $j=k+1$. Then we apply \eqref{eq:74} with $m=2$,
$p_1<\tilde{p}/n$,  $p_2 =q$,  $f_1(s)=s^{-n/\tilde{p}}$,  $f_2=
\dot h$ yielding an extra factor
$s^{n/\tilde{p}}$. Upon
using  \eqref{eq:12tt} for the remaining $k$
factors of components of  $h_{\bullet}$'s, \eqref{eq:9pp}  and the fact that
\begin{equation*}
  s^{-|\eta|}s^{|\zeta_1|}s^{k\parub{1-1/{\tilde p}}}s^{n/\tilde p}\leq 1,
\end{equation*} we conclude as before.

\myparagraph{Case Aiii):} $\partial^\eta
G(\gamma(s))=s^{-|\zeta_2|}\partial^{\zeta_2}_x\partial_{\kappa(s)}^{\omega}
G(sx+\kappa(s))$; $|\zeta_2|=|\zeta|-2,\;|\omega|=k+1$. We apply \eqref{eq:74} with $m=2$,
$p_1<\tilde{p}/n$, $p_2=q$, $f_1(s)=s^{-n/\tilde{p}}$   and $f_2(s)= h(s)/s$ yielding an extra factor
$s^{1+n/\tilde{p}}$. Upon
using  \eqref{eq:12tt} for the remaining $k$
factors of components of  $h_{\bullet}$'s, \eqref{eq:9pp}  and the fact that
\begin{equation*}
  s^{-|\eta|}s^{|\zeta_2|}s^{k\parub{1-1/{\tilde p}}}s^{1+n/\tilde p}\leq 1,
\end{equation*} we  conclude as before.

For the case \textbf{B)} we need \eqref{eq:9pp} with
$|\eta|=|\zeta|+k$, and we consider subcases
\textbf{Bi)}--\textbf{Bii)} defined as before.

\myparagraph{Case Bi):}
$\partial^\eta
G(\gamma(s))=s^{-|\zeta|}\partial^\zeta_x\partial_{\kappa(s)}^{\omega}
G(sx+\kappa(s));\;|\omega|=k$.  Suppose first
that
$j\leq k$.
Then we apply \eqref{eq:74} with $m=3$,
$p_1<\tilde{p}/n$, $p_2 =\tilde{p}$ and $p_3=q$, $f_1(s)=s^{-n/\tilde{p}}$, $f_2=
\dot h_j$ and $f_3(s)= h(s)/s$ yielding an extra factor
$s^{1+n/\tilde{p}}$.
Upon
using  \eqref{eq:12tt} for the remaining $k-1$
factors of components of $h_{\bullet}$'s, \eqref{eq:9pp}  and the fact that
\begin{equation*}
  s^{-|\eta|}s^{|\zeta|}s^{(k-1)\parub{1-1/{\tilde p}}}s^{1+n/\tilde p}\leq 1,
\end{equation*} we obtain
the bound
\begin{equation}
  \label{eq:13bbi}
  |(F_{\zeta,k};h_1,\dots, h_{k+1})|\leq
  C
  |x|^1|x|^{-|\eta|}\prod_{m=1}^{k} \|h_m\||_{\vH^{\tilde p}}\,\|h\|_{\vH^q}.
\end{equation} By the induction hypothesis
\begin{equation*}
  \prod_{m=1}^{k} \|h_m\|_{\vH^{\tilde p}}\leq C\inp{x}^{k-\sum|\beta_m|}=C\inp{x}^{k-(n-|\zeta|)},
\end{equation*} which together with \eqref{eq:13bbi} yields
\begin{equation}
  \label{eq:14bbi}
  |(F_{\zeta,k};h_1,\dots, h_{k+1})|\leq
  C
  \inp{x}^{1-|\eta|}\inp{x}^{k-(n-|\zeta|)}\|h\|_{\vH^q}= C\inp{x}^{1-n}\|h\|_{\vH^q}.
\end{equation}

Suppose next that $j=k+1$. Then we apply \eqref{eq:74} with $m=2$,
$p_1<\tilde{p}/n$,  $p_2 =q$,  $f_1(s)=s^{-n/\tilde{p}}$,  $f_2=
\dot h$ yielding an extra factor
$s^{n/\tilde{p}}$. Upon
using  \eqref{eq:12tt} for the remaining $k$
factors of components of  $h_{\bullet}$'s, \eqref{eq:9pp}  and the fact that
\begin{equation*}
  s^{-|\eta|}s^{|\zeta|}s^{k\parub{1-1/{\tilde p}}}s^{n/\tilde p}\leq 1,
\end{equation*} we conclude as before.

\myparagraph{Case Bii):}
$\partial^\eta
G(\gamma(s))=s^{-|\zeta_1|}\partial^{\zeta_1}_x\partial_{\kappa(s)}^{\omega}
G(sx+\kappa(s));\;|\zeta_1|=|\zeta|-1,\;|\omega|=k+1$. We apply \eqref{eq:74} with $m=2$,
$p_1<\tilde{p}/n$, $p_2=q$, $f_1(s)=s^{-n/\tilde{p}}$   and $f_2(s)= h(s)/s$ yielding an extra factor
$s^{1+n/\tilde{p}}$. Upon
using  \eqref{eq:12tt} for the remaining $k$
factors of components of  $h_{\bullet}$'s, \eqref{eq:9pp}  and the fact that
\begin{equation*}
  s^{-|\eta|}s^{|\zeta_1|}s^{k\parub{1-1/{\tilde p}}}s^{1+n/\tilde p}\leq 1,
\end{equation*} we  conclude as before.

For the case   C) we  need  \eqref{eq:9pp} with
$|\eta|=|\zeta|+k+1$.

\myparagraph{Case C):}
 $\partial^\eta
G(\gamma(s))=s^{-|\zeta|}\partial^\zeta_x\partial_{\kappa(s)}^{\omega}
G(sx+\kappa(s));\;|\omega|=k+1$.

We apply \eqref{eq:74} with $m=2$, $p_1<\tilde{p}/n$, $p_2=q$,
$f_1(s)=s^{-n/\tilde{p}}$ and $f_2(s)= h(s)/s$ yielding an extra
factor $s^{1+n/\tilde{p}}$. Upon using \eqref{eq:12tt} for the
remaining $k$ factors of components of $h_{\bullet}$'s, \eqref{eq:9pp}
and the fact that
\begin{equation*}
  s^{-|\eta|}s^{|\zeta|}s^{k\parub{1-1/{\tilde p}}}s^{1+n/\tilde p}\leq 1,
\end{equation*}  we obtain
the bound
\begin{equation}
  \label{eq:13bbci}
  |(F_{\zeta,k};h_1,\dots, h_{k+1})|\leq
  C
  |x|^2|x|^{-|\eta|}\prod_{m=1}^{k} \|h_m\||_{\vH^{\tilde p}}\,\|h\|_{\vH^q}.
\end{equation} By the induction hypothesis
\begin{equation*}
  \prod_{m=1}^{k} \|h_m\|_{\vH^{\tilde p}}\leq C\inp{x}^{k-\sum|\beta_m|}=C\inp{x}^{k-(n-|\zeta|)},
\end{equation*} which together with \eqref{eq:13bbci} yields
\begin{equation}
  \label{eq:14bbci}
  |(F_{\zeta,k};h_1,\dots, h_{k+1})|\leq
  C
  \inp{x}^{2-|\eta|}\inp{x}^{k-(n-|\zeta|)}\|h\|_{\vH^q}=
  C\inp{x}^{1-n}\|h\|_{\vH^q}.
  \qedhere
\end{equation}
 \end{proof}

  \begin{remarks*}
    \begin{enumerate}[1)]
    \item \label{item:13} As noticed in the beginning of the above
      proof, using the proof of Theorem \ref{thm:main result}
      \ref{item:2}) we can improve \eqref{eq:72b} for $|\alpha|=1$ as
      \begin{equation}
        \label{eq:6iill}
        \|\partial_x^\alpha\kappa\|_{\vH^p}\leq \sqrt \epsilon C_p\mfor |\alpha|= 1.
      \end{equation}

    \item \label{item:14}By integrating \eqref{eq:6iill} we obtain the
      bound
      \begin{equation}
      \label{eq:72bdd}
      \|\kappa\|_{\vH^p}\leq \sqrt \epsilon C_p\inp{x}.
    \end{equation} 
    From the method of proof we have $C_p\rightarrow \infty$ as
    $p\rightarrow \infty$.  We remark that this feature is not an
    artifact of the proof. In fact for the example in Subsection~\ref{Another example} the geodesics emanating from $0$ are
    rotating like logarithmic spirals showing that \eqref{eq:72bdd} is
    false for $p=\infty$ (here by definition
    $\vH^\infty=W_0^{1,\infty}(0,1)^d$). This partly explains why we
    worked with \eqref{eq:69} rather than the $\epsilon$-independent
    pairing \eqref{eq:35nb}.
    \end{enumerate}
\end{remarks*}

\begin{proof}[Proof of Theorem \ref{thm:main-result2}.]
We mimic the proof of Proposition \ref{prop:soleikeq}
\ref{item:5}). Indeed combining  \eqref{eq:5} and the bounds
\eqref{eq:72b} and \eqref{eq:72bdd} (with $p=2$) we show the following improvement of  \eqref{eq:16}
\begin{equation*}
  |\partial^\alpha\dot\kappa(1)|\leq C\langle x\rangle^{1-|\alpha|}\mforall |\alpha|\leq l-1.
\end{equation*} Next we use again  \eqref{eq:2} and obtain
\eqref{eq:1bhhhh}. Obviously we can similarly  use   \eqref{eq:2} to prove \eqref{eq:1bhhhhg}.
\end{proof}

\section{Examples} \label{Examples}

In this section we present examples  of metrics
in the class $\mathcal O$. Two of  the examples are
parameter-depending and are constructed   by the exponential mapping from metrics that are not of
order zero.

\subsection{Decaying potentials} \label{PrinExamples}

Let  $V$ be a  radial function of class $C^l$  on $\R^d$, $l,d\geq 2$,   for which
 there are constants $a>0$, $A>0$, $ \mu\in ]0,2[$ and $\sigma\in
 ]0,2]$ such that
 \begin{subequations}
   \begin{align}
\label{bound1v}
-A \langle x\rangle^{-\mu}\leq V(x)&\leq -a \langle x\rangle^{-\mu},
\\
x\cdot \nabla V(x)+2V(x)&\leq \sigma V(x)\label{bound1vbb},\\
\partial^\alpha V(x)&=O\parb{\inp{x}^{-(\mu+|\alpha|)}}\mfor |\alpha|\leq l.\label{bound1vbbcc}
\end{align}
 \end{subequations}
  Consider the functional $J:\R^d\times\mathcal H\rightarrow \R$
given  by
 \begin{equation}
\label{metrg1}
J(x,\kappa)=\int_0^1K(y(s))|\dot y(s)|^2ds,
\end{equation}
where $y(s)=sx+\kappa(s)$ and
 $K(x)=2(\lambda-V(x))$ for  $\lambda \geq 0$,  and  consider the positive solution $S(x)$ to the eikonal equation
 \begin{equation}\label{eq:49}
|\nabla S(x)|^2=K(x) \mfor x\in \R^d\setminus\{0\},
\end{equation}
defined  by
\begin{equation*}
S(r)=\int_0^r\sqrt{K(\tau)}d\tau,
\end{equation*} or alternatively by
 \begin{equation*}
S^2(x)=\inf\{J(x,\kappa):\kappa\in \mathcal H\}.
\end{equation*}
 We introduce  the diffeomorphism
$\Phi:\R^d\rightarrow\R^d$ given by
 \begin{equation}
\label{transft}
\Phi(x)=r(|x|)\hat x=r(|x|)\frac{x}{|x|},
\end{equation}where $r$ is the inverse of the function
 \begin{equation*}
\rho(r)=S(r)=\int_0^r\sqrt{K(\tau)}d\tau.
\end{equation*} Note that $\Phi(t\sqrt{K(0)}x)=\exp_0(tx)$, i.e. the
exponential mapping at zero for the metric $g_{ij}=K\delta_{ij}$ of \eqref{metrg1}.
 A short calculation  shows that if
 $z(s)=\Phi(y(s))$ and $\omega(s)=y(s)/|y(s)|$ then
\begin{align*}
\dot z&=\frac{1}{\sqrt{K(r(|y|))}}\dot y\cdot\omega\omega+\frac{r(|y|)}{|y|}(\dot y-\dot y\cdot\omega\omega)\\
&= \frac{1}{\sqrt{K(r(|y|))}}P\dot y+\frac{r(|y|)}{|y|}P_\perp\dot y
\end{align*}
 where $P$ is the projection parallel to $\omega$ and $P_\perp=I-P$ the
  projection onto $\{\omega\}^\perp$. Therefore
\begin{align*}
K(z)|\dot z|^2&= |P\dot y|^2+f^2(|y|)|P_\perp\dot y|^2\\
&=\dot y G(y)\dot y
\end{align*}
where
\begin{subequations}
 \begin{equation}
\label{defg0}
G(y)=P+f^2(|y|)P_\perp
\end{equation}
and
  \begin{equation}\label{defg0bb}
    f(\rho)=\frac{\sqrt{K(r(\rho))} r(\rho)}{\rho}=\biggl(
    {\int_0^1\sqrt{\frac{\lambda-V(sr(\rho))}
        {\lambda-V(r(\rho))}}ds}\biggr)^{-1}. 
\end{equation}
\end{subequations}

Thus, if $y(s)=sx+\kappa(s)$ and $z(s)=\Phi(y(s))=s\Phi(x)+h(s)$,  with
  $\kappa, h\in\mathcal H$,    then
  \begin{equation*}
J(\Phi(x),h)=\int_0^1\dot y(s)G(y(s))\dot y(s) ds=:E(x,\kappa)
\end{equation*}
and  therefore
 \begin{equation*}
\inf\{E(x,\kappa):\kappa\in \mathcal
H\}=\inf\{J(\Phi(x),h):h\in\mathcal H\}=S^2(\Phi(x)).
\end{equation*}
Now we show that  $G$ satisfies \eqref{ellcond} with   constants $a$ and $b$ independent of $\lambda\geq 0$.
It suffices to prove  there exist $c>0$  and  $C>0$ such that
  \begin{equation*}
c\leq\int_0^1\sqrt{\frac{\lambda-V(sr(\rho))} {\lambda-V(r(\rho))}}ds\leq C,
\end{equation*}
for all $\lambda\geq 0$ and $r\geq 0$.
Note first that from \eqref{bound1v} we have
  \begin{equation*}
\frac{\lambda+a\langle sr\rangle^{-\mu}} {\lambda+A\langle r\rangle^{-\mu}}
\leq\frac{\lambda-V(sr)} {\lambda-V(r)}
\leq \frac{\lambda +A\langle sr\rangle^{-\mu}} {\lambda+ a\langle r\rangle^{-\mu}}
\end{equation*}
The lower bound follow from the fact  that
\begin{subequations}
 \begin{equation}\label{eq:51}
\frac{\lambda+a\langle sr\rangle^{-\mu}} {\lambda+A\langle r\rangle^{-\mu}}\geq
\frac{\lambda+a\langle r\rangle^{-\mu}} {\lambda+A\langle r\rangle^{-\mu}}\geq{\frac{a}{A}},
\end{equation}
for  $\lambda\geq 0$, $r\geq 0$ and $s\in[0,1]$.
Next we  note that
  \begin{equation}\label{eq:52}
\frac{\lambda +A\langle sr\rangle^{-\mu}} {\lambda+ a\langle r\rangle^{-\mu}}\leq \frac{A}{a}
\frac{\langle r\rangle^{\mu}}{\langle sr\rangle^{\mu}},
\end{equation}
\end{subequations}
  and therefore
  \begin{align*}
\int_0^1\sqrt{\frac{\lambda +A\langle sr\rangle^{-\mu}} {\lambda+ a\langle r\rangle^{-\mu}}} ds
&\leq  \sqrt{ \frac{A}{a}}\langle r\rangle^{\mu/2}
\int_0^1\frac{ds}{(1+s^2r^2)^{\mu/4}}\\
&= \sqrt{ \frac{A}{a}}\frac{\langle r\rangle^{\mu/2}}{r}\int_0^r\frac{du}{(1+u^2)^{\mu/4}}.
\end{align*}
The upper bound  is obtained  from the fact that the  function
  \begin{equation*}
\phi(r)=\frac{\langle r\rangle^{\mu/2}}{r}\int_0^r\frac{du}{(1+u^2)^{\mu/4}}
\end{equation*}
is continuous for $r>0$ and satisfies
  \begin{equation*}
\lim_{r\to 0}\phi(r)=1,\qquad\text{and}\qquad \lim_{r\to\infty}\phi(r)=2/(2-\mu).
\end{equation*}

Thus,  since it can easily be verified that $G$ satisfies
\eqref{condoz}, it follows  that indeed $G$
satisfies \eqref{ellcond}. From from \eqref{defg0} it follows that $G$
satisfies  \eqref{ortdec}. A computation using \eqref{bound1vbb} shows
the bound  \eqref{ortdec2} with any $\bar c\leq \tfrac \sigma2 /\sup f$,
whence $G\in \mathcal O$.

Consequently, due to  Theorem \ref{thm:main result},  if a real $C^l$-potential on $\R^d$, say $ V_\epsilon$, is
sufficiently close to the  radial potential $V$  discussed above in the
sense that for a sufficiently small $\epsilon>0$
\begin{equation}
  \label{eq:48}
  \| V_\epsilon-V\|_l:=\sup_{ |\alpha|\leq l}\sup_x\,\inp{x}^{|\alpha|+\mu}|\partial ^\alpha\parb{ V_\epsilon(x)-
    V(|x|)}|\leq \epsilon,
\end{equation}
then there exists a $C^l$-solution $ S_\epsilon$ to the eikonal
equation \eqref{eq:49} with $V\to V_\epsilon$.  Indeed writing the
solution from Theorem \ref{thm:main result} defined in terms of the
perturbed metric, say~$G_\epsilon$, as $\tilde S_\epsilon$ (in the
coordinates $y=\Phi^{-1}(z)$) we have $S_\epsilon(x)=\tilde
S_\epsilon(\Phi^{-1}(x))$ and we can use the estimates of Theorem
\ref{thm:main result} to compare the derivatives $\partial^\alpha
S_\epsilon$ of order $|\alpha|\leq 2$ with the corresponding
unperturbed quantities $\partial^\alpha S$. The comparison in mind is
given in terms of estimates exhibiting smallness in terms of the
parameter $\epsilon>0$. If we (for simplicity) assume that $V$ is
constant in a neighbourhood of zero then it is straightforward to show
that these estimates are uniform not only in $|x|\geq r$ for any $r>0$
but also in $\lambda \geq0$. All that is needed to show at this point
is, cf. Remark \ref{remark:cond-main-results} \ref{item:10a}), that
\begin{equation}
  \label{eq:53}
 \sup_{\lambda\geq0}\|G_\epsilon-G\|_l\to 0\text{ as } \|V_\epsilon  -V\|_l\to 0.
\end{equation}
  For \eqref{eq:53}  we use in turn
(and note for comparison) that the estimates of the
constants in \eqref{ellcond} and \eqref{ortdec2} given above are uniform
in $\lambda\geq 0$. The constructed  solution $S_\epsilon$
has an  application in scattering theory at low energies  i.e. in
the regime $\lambda\to 0$ generalizing parts of \cite{DS}, see
\cite{Sk}.
\subsection{Non-decaying  potentials} \label{Rel Examples
 }Let  $V$ be a  radial function of class $C^l$  on $\R^d$, $l,d\geq 2$,   for which
 there are constants $a>0$, $A>0$, $ \mu\in ]-\infty ,0]$ and $\sigma\in
 ]0,2]$ such that
 \begin{subequations}
   \begin{align}
\label{bound1vb}
-A \langle x\rangle^{-\mu}\leq V(x)&\leq -a \langle x\rangle^{-\mu},
\\
x\cdot \nabla V(x)+2V(x)&\leq \sigma V(x)\label{bound1vbbb},\\
\partial^\alpha V(x)&=O\parb{\inp{x}^{-(\mu+|\alpha|)}}\mfor |\alpha|\leq l.\label{bound1vbbccb}
\end{align}
 \end{subequations} Notice that the these conditions are very similar
 to 
 \eqref{bound1v}--\eqref{bound1vbbcc}. The only difference is that now
 $\mu\leq 0$. We can again allow the inclusion of a positive energy,
 $\lambda\geq 0$ and obtain uniform estimates. In fact in this case
 \eqref{bound1vb}--\eqref{bound1vbbccb} are invariant under the
 replacement $V\to V-\lambda$ (up to a change of constants), so if not for the uniformity in
 $\lambda \geq0$ we could in the following  take
 $\lambda=0$. We get the bounds on $f$ in a similar fashion. Note that
 we need to replace \eqref{eq:51} and \eqref{eq:52} by
 \begin{subequations}
  \begin{equation}\label{eq:51b}
\frac{\lambda+a\langle sr\rangle^{-\mu}} {\lambda+A\langle
  r\rangle^{-\mu}}\geq{\frac{a}{A}}\frac{\langle sr\rangle^{-\mu}}
{\langle r\rangle^{-\mu}},
\end{equation} and
  \begin{equation}\label{eq:52b}
\frac{\lambda +A\langle sr\rangle^{-\mu}} {\lambda+ a\langle r\rangle^{-\mu}}\leq \frac{\lambda +A\langle r\rangle^{-\mu}} {\lambda+ a\langle r\rangle^{-\mu}}\leq\frac{A}{a},
\end{equation}
 \end{subequations} respectively. Again these estimates are for
$\lambda\geq 0$, $r\geq 0$ and $s\in[0,1]$. We take the square root in
\eqref{eq:51b} and \eqref{eq:52b} and integrate. As for \eqref{eq:51b}
we then use a change of variables and obtain
\begin{equation*}
  \int_0^1\sqrt{\frac{\lambda-V(sr(\rho))} {\lambda-V(r(\rho))}}ds\geq
  \sqrt {\frac{a}{A}}\langle r\rangle^{\mu/2} r^{-1}\int_0^r\langle
  u\rangle^{-\mu/2}du\geq c.
\end{equation*} Obviously from \eqref{eq:52b} we obtain
\begin{equation*}
  \int_0^1\sqrt{\frac{\lambda-V(sr(\rho))} {\lambda-V(r(\rho))}}ds\leq
  \sqrt {\frac{A}{a}}\leq C.
\end{equation*} We have verified that $G$ given by \eqref{defg0} is a
metric of order zero. As before  we verify the bound  \eqref{ortdec2} with any $\bar c\leq \tfrac \sigma2  /\sup f$,
whence $G\in \mathcal O$. From this point we can proceed as before and
introduce  a class of perturbations $V_\epsilon$ by \eqref{eq:48}
(now with $\mu\leq 0$) and indeed show the existence of a $C^l$-solution $ S_\epsilon$ to the eikonal
equation \eqref{eq:49} with $V\to  V_\epsilon$.

We remark that the class of metrics  discussed above by perturbing $V(x)=-1/2$
(in the particular case $\mu=\lambda=0$) coincides with the class considered in
\cite{Ba} (here we also take    $l=3$). The parameter $\epsilon>0$ may
play the role of inverse
energy. In particular the constructed solution to the eikonal
equation  was applied in \cite{ACH} in the study of scattering theory
of
Schr\"odinger operators with an order zero potential in the high energy regime.

\subsection{Order zero potential, logarithm orbits} \label{Another example}
We demonstrate in terms of an example from \cite{HS} (see \cite[Example
A.6]{HS}) that perturbations of the Euclidean metric in the sense of
this paper may involve
somewhat exotic geodesics like logarithm orbits. This means that
although the direction of any geodesic emanating from $0$, say
$\gamma$, and its velocity $\dot\gamma$ are almost ligned up for all
times, cf. Lemma \ref{lemma:2_perturbed-case} \ref{item:7}), the
geodesic is permanently rotating around  $0$ as $|\gamma|\to \infty$.

Consider the symbol $h$ on ${\R}^{2} \times {\R}^{2}$ given by
  $h=h(x,\xi)=g^{-1} \xi^2$,
  where the conformal (inverse) metric factor is
 specified in polar coordinates
 $x=(r\cos \theta, r\sin \theta)$ as
 $g^{-1}=e^{-2\epsilon f_\epsilon\chi};\;f_\epsilon=f(\theta-\epsilon\ln r)$, $\chi=\chi(r>1)$. We assume $f$ is a given smooth
 $2\pi$--periodic function with $\max f'\geq 1$ and that
 $\epsilon>0$ is small. Notice that the ``$x$-space part'' of the Hamiltonian orbits of this symbol are
 the geodesics of the metric $g_{ij}=g\delta_{ij}$.
Consider the orbits  originating at  $(r_0,0;C,\epsilon C)$ where
$C>0$ is arbitrary and $r_0>2$ is determined by the
  equation
  \begin{equation}\label{eq:50}
 f'(\theta_0) =(1+\epsilon^2)^{-1};\;\theta_0=-\epsilon\ln r_0.
  \end{equation}  By assumption there
  exists at least one such solution. Let us assume that there are only
  a finite number of such solutions, say $\theta_j$,  all being non-degenerate,
  $f''(\theta_j )\neq 0$. The $x$-space part of any corresponding  orbit  is the logarithmic spiral
 given by the equation $\theta-\epsilon\ln r= \theta_0$. For
 $f''(\theta_0) >0$ the orbit corresponds in reduced variables to a saddle, see
 \cite{HS}. On the other hand for
 $f''(\theta_0) <0$ the orbit corresponds to a sink. This means that
 generically  the geodesics of the metric $g_{ij}=g\delta_{ij}$  emanating from $0$ are
 attracted to one of the logarithmic spirals associated to
 \eqref{eq:50} and the
 condition $f''(\theta_0) <0$.

\end{document}